\documentclass[pra,aps,twocolumn,superscriptaddress,groupedaddress,nofootinbib]{revtex4-1}
\usepackage{amsmath,amsfonts,amsthm,bbm,graphicx,enumerate,hyperref,xcolor,tikz,pgfplots,ulem,subfigure}
\usepackage{comment}
\usepackage{multirow}
\usepackage[export]{adjustbox}
\usepackage{pifont}
\usepackage{placeins}

\newcommand{\cmark}{\ding{51}}%
\newcommand{\xmark}{\ding{55}}%

\hypersetup{
        colorlinks=true,  
        linkcolor=blue,   
        citecolor=blue,   
        urlcolor=blue     
}

\def\E{ {\mathcal E} }
\def\P{ {\mathcal P} }

\def\U{ {\mathcal U} }

\def\N{ {\mathcal N} }

\def\C{ {\mathcal C} }
\def\O{ {\mathcal O} }
\def\P{ {\mathcal P} }

\def\S{ {\mathcal S} }
\def\T{ {\mathcal T} }
\def\V{ {\mathcal V} }
\def\I{ {\mathcal I} }

\def\Tr{{\rm{tr}}}
\def\tr{ \mbox{tr} }

\def\>{\rangle}
\def\<{\langle}

\def\hc{^{\dagger}}

\renewcommand{\emph}{\textit}

\newcommand{\iden}{\mathbb{I}}

\newtheorem{theorem}{Theorem}

\newtheorem{lemma}[theorem]{Lemma}
\newtheorem{definition}[theorem]{Definition}

\usepackage{microtype}

\begin{document}

\title{Spectral analysis for noise diagnostics and filter-based digital error mitigation}

\author{Enrico Fontana$^{1,2,3}$}
\author{Ivan Rungger$^3$}
\author{Ross Duncan$^{1,2,4}$}
\author{Cristina C\^{i}rstoiu$^{1}$}

\affiliation{$^1$Cambridge Quantum Computing Ltd, Terrington House, 13-15
  Hills Road, Cambridge CB2 1NL, UK}
\affiliation{$^2$Department of Computer and Information Sciences, University of Strathclyde, 26 Richmond Street, Glasgow G1 1XH, UK}
\affiliation{$^3$National Physical Laboratory, Hampton Road, Teddington TW11 0LW, UK}
\affiliation{$^4$Department of Physics and Astronomy, UCL, Gower Street, London, WC1E 6BT, UK}

\begin{abstract}
        We investigate the effects of noise on parameterised quantum
        circuits using spectral analysis and classical signal processing tools.
        For different noise models, we
        quantify the additional, higher frequency modes in
        the output signal caused by device errors.
        We show that filtering these noise-induced
        modes effectively mitigates device errors. When
        combined with existing methods, this yields an improved
        reconstruction of the noiseless variational landscape.
        Moreover, we describe the classical and quantum resource
        requirements for these techniques and test their effectiveness
        for application motivated circuits on quantum hardware.
\end{abstract}

\maketitle

\section{Introduction}
\label{sec:introduction}

Characterisation and verification of quantum computations becomes increasingly important as quantum hardware scale and facilitate noisy intermediate scale quantum (NISQ) applications for larger problem sizes. In most cases, there is a trade-off between the amount of information obtained about the errors and resource requirements of the protocols that acquire it~\cite{eisert2020quantum}. 

Much progress has been made in quantifying and reducing errors of individual components by using methods specific to modelled noise sources in a given physical implementation (e.g spectroscopy, pulse sequences or dynamical decoupling) and by developing hardware-agnostic tools for noise characterisation (e.g tomographic techniques \cite{merkel2013self, blume2013robust, flammia2020efficient, bonet2020nearly, cotler2020quantum}, randomised benchmarking protocols \cite{emerson2005scalable, knill2008randomized, cross2016scalable, carignan2015characterizing}).

For larger circuits, process tomography quickly becomes prohibitively expensive, so one must resort to holistic benchmarks \cite{boixo2018characterizing,cross2019validating, mills2020application}.  Typically these produce a single metric to characterise the performance of a quantum device at running application-motivated circuits of specified sizes \cite{proctor2022measuring}.  However, the question of how to quantify the effect of noise for given circuit structures remains.  Tools that allow to gain this additional information to an extent that can inform strategies for noise suppression at the circuit level will be needed to push NISQ applications towards reaching regimes where computational quantum advantage may be possible  \cite{kim2021scalable}. 

Near-term (variational) algorithms typically involve circuits with a set of variable gate parameters, but otherwise a fixed structure. Although variational methods grant partial resilience to noise \cite{fontana2021evaluating, sharma2020noise}, accumulation of errors still impacts their ability to generate expectation values that are competitive with classical methods \cite{stilck2021limitations}. Furthermore, noise can also impact the trainability of variational algorithms through the presence of barren plateaus in the optimisation landscape for increasing problem sizes \cite{wang2021noise}. 

Error mitigation methods aim to reduce the approximation error in
the noisy estimation of expectation values and are generally designed for
near-term applications.
While they typically require a sampling overhead that scales
exponentially (with depth)
\cite{takagi2021optimal,takagi2021fundamental}, error mitigation can
be achieved in regimes beyond
classical simulability \cite{kim2021scalable,lostaglio2021error,piveteau2021error}. 

Some of the techniques involve boosting errors and extrapolation to zero noise levels \cite{temme2017error, ying2017efficient},  using quasiprobabilities to decompose noisy gates \cite{temme2017error, endo2018practical}, machine learning approaches \cite{strikis2020learning, czarnik2020error} and combining methods \cite{cai2021multi,lowe2021unified}. However, in regimes with low quantum volume \cite{cross2019validating}, the general performance of different error mitigation methods can exhibit a great deal of variability on physical hardware \cite{cirstoiu2022volumetric}.
Combining noise suppression at circuit and gate level with error mitigation and knowledge of the specific circuit or application may be a way to circumvent this  \cite{kim2021scalable}. 

In this work we investigate the effects of noise on parameterised
quantum circuits with a fixed circuit structure by means of signal
processing methods. We employ harmonic analysis at the level of
quantum channels, along with decompositions of unitaries into linear
combinations of Clifford operations. This allows us to relate
expectation values of parameterised circuits to trigonometric
polynomials. Previous works also highlighted this relationship for specific
applications \cite{schuld2021effect, vidal2018calculus, gil2020input, nakanishi2020sequential, parrish2019jacobi, ostaszewski2021structure, wierichs2022general}, and
our contribution expands and generalises the connection along several
directions. We characterise both the frequency
support, in terms of lattice structures generated by parameters of the
circuit, and the coefficients in terms of process modes that are
derived through channel decompositions into Clifford paths. These
observations allow us to treat noisy expected values as output signals
that carry information about how errors affect the specified circuit.

We analyse the effects of noise in the frequency domain by taking
(fast) Fourier transforms with respect to the variational
parameters.
In particular, for real hardware implementations the
presence of frequencies outside the support of the noiseless input
signal indicates not only the gate dependent errors (due to
calibration or control imprecision) but also how these build up
through the circuit.
Therefore, these additional modes also pick up information about the
correlated errors induced by the circuit structure.
Furthermore, we explore the effect of common noise models: for instance,
incoherent noise leads to contractions of Fourier
coefficients. 

Crucially, for a target set of parametrised circuits and observables the harmonic analysis infers the support of the frequency spectrum. This additional information certifies that the presence of modes outside this support is attributed to noise alone, a connection that allows us to design new error mitigation methods that target these isolated noise components.
These spectral filtering techniques can be shown to work exactly (below a fixed noise threshold) for certain classes of classically simulable circuits. We also demonstrate their effectiveness for practically relevant circuit structures used for quantum chemistry and combinatorial optimisation problems, by combining them with complementary error mitigation methods targeting decoherence, and demonstrate improved estimators for expectation values over the whole variational landscape.

\begin{figure}[h!]
        \includegraphics[width=0.4\textwidth]{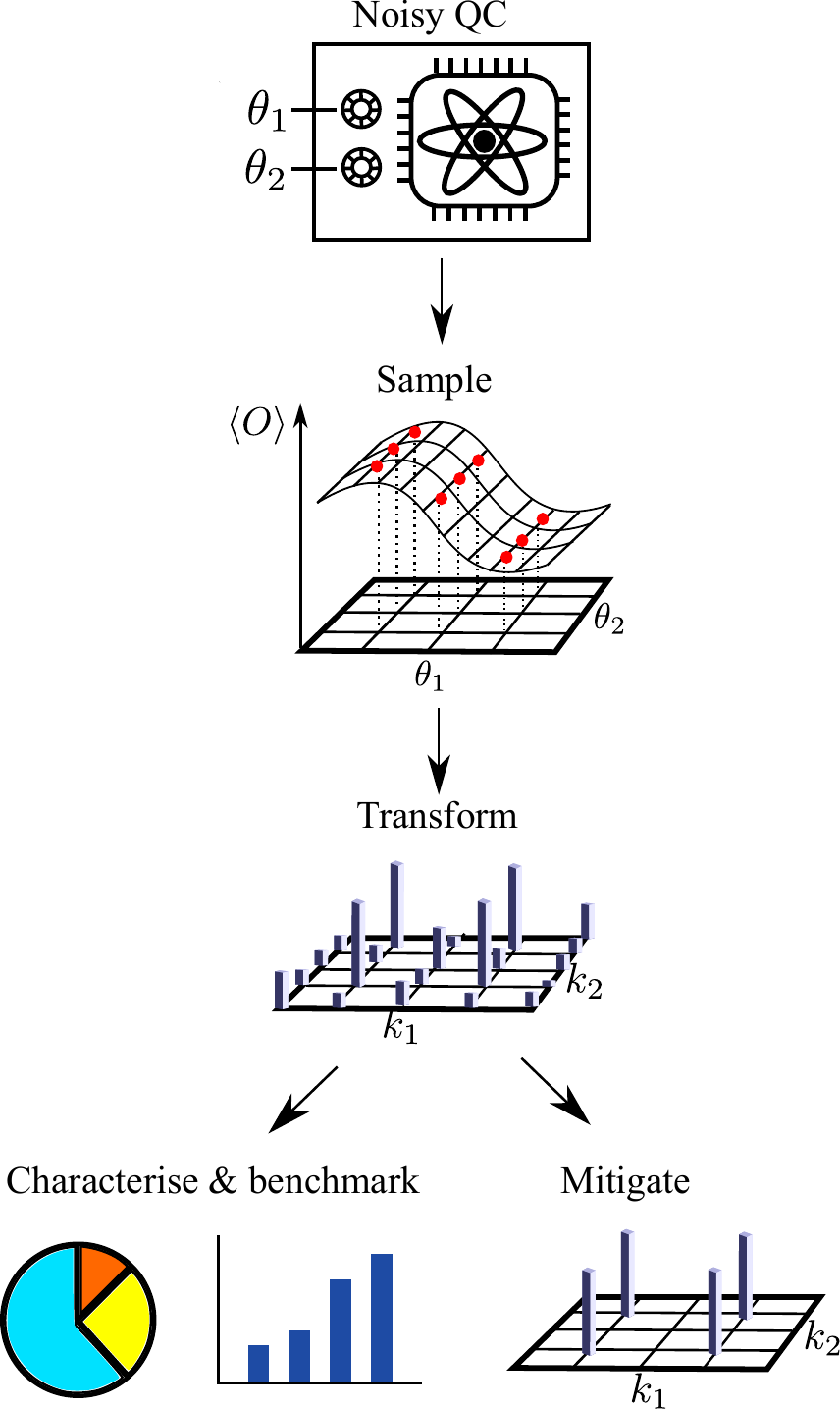}
        \caption{Schematic overview of the techniques presented in the paper. The output of a measurement operator on noisy parameterised quantum circuit is sampled at different parameter levels; the discrete samples are transformed (e.g. with FFT) to deduce the frequency spectrum of the noisy landscape; the spectral information is then used to characterise and benchmark the machine or to reduce the noise to give a mitigated cost function landscape.}
        \label{fig:cartoon}
\end{figure}

\section{Expectation values of parameterised states as trigonometric polynomials}
\label{sec:Decomposition}

Several authors have observed the simple relation between expectation values of
observables with respect to parameterised circuits and trigonometric
functions, or finite Fourier series \cite{schuld2021effect, vidal2018calculus, gil2020input, nakanishi2020sequential, parrish2019jacobi, ostaszewski2021structure, wierichs2022general}. In particular
Schuld et al.\!~\cite{schuld2021effect} give a general analysis
that is used to explore the expressivity of ans\"atze for
data-encoding.
The particular circuit involved repeated applications of the same parameterised unitary of the form $V:=V_0 U(\theta)V_1 U(\theta) ...  U(\theta) V_m$ with one-parameter unitaries  $U(\theta)$ generated by a Hermitian operator $H$ acting on the Hilbert space corresponding to $n$ qubits.

The key idea is that if $H$ has an eigenvalue decomposition
$H= W\sum_{z=0}^{2^{n}-1} \lambda_z |z\>\<z| W\hc$
with $|z\>$ computational basis states, then
$U(\theta) =We^{i \theta \sum_z \lambda_z |z\>\<z|}W\hc$,
and therefore the expectation value of a target observable $O$ is given by 
\begin{align}
  \Tr(OU(\theta)\rho U(\theta)\hc)
  &= \Tr(\tilde{O}e^{i\theta \sum_z\lambda_z |z\>\<z|}
     \tilde{\rho} e^{-i\theta \sum_z \lambda_z |z\>\<z|} ) \nonumber\\
  &= \sum_{z,y} e^{i\theta(\lambda_z-\lambda_y)}
      \tilde{O}_{yz}\tilde{\rho}_{zy}
 \label{eq:eigenvalue-expansion}
\end{align}
where $\tilde{O}_{yz} := (WOW\hc)_{yz}$ and $\tilde{\rho}_{yz} :=
(W\rho W\hc)_{yz}$ are the matrix entries of the observable and input
state in a rotated frame. 

We extend these results to generalised circuits with $m$ independent
parameters. Our approach is conceptually different, as we decompose
the target unitary operation into a linear combination of operations
that separates the dependency on the parametrisation from the specific
circuit structures involved. This is related to sum-over-Clifford
paths approaches to classical simulation of circuits
\cite{bravyi2019simulation, huang2021feynman}.  This channel-based
perspective is particularly well-suited to analyse the effects
of noise.

We consider unitaries that are parameterised by a set of angles
$\boldsymbol{\theta} = [\theta_1, ..., \theta_m]$
so that they take the general form 
\begin{align}
  U(\boldsymbol{\theta})
  = V_0 U_1(\theta_1) V_1 U_2(\theta_2) V_2 \, ...\, U_m(\theta_m)V_m.
  \label{eq:unitary_circuit}
\end{align}
where each of the one-parameter unitaries are generated by hermitian
operators $U(\theta_j) = e^{i\theta_j H_j }$.  

An $m$-variate trigonometric polynomial with degree at most $d$ has the form
\begin{equation}
        \label{eq:trigPoly}
        f(\boldsymbol{\theta}) = \sum_{k_1, k_2,..., k_m =-d}^{d} c_{\bf{k}} e^{i \mathbf{k}\cdot \boldsymbol{\theta}}.
\end{equation}
where $\mathbf{k} = (k_1, ..., k_m)$ and $k_i\in \mathbb{Z}$, $\theta_i \in [0, 2\pi]$ and ${\bf{k}}
\cdot \boldsymbol{\theta} = \sum_{i=1}^m k_i \theta_i$.
Furthermore, for real functions $f$ it holds that $c_{-\bf{k}} =
c_{\bf{k}}^{*}$. We will denote by $\mathcal{T}_{m,d}$ the set of all
such real trigonometric functions on $m$ variables with degree at most
$d$; $\bf{k}$ is called a \emph{path}.

\subsection{Illustrative example: Pauli gadgets}
\label{sec:Example}
Consider first a single-qubit parameterised $Z$ rotation gate
$Z(\theta_j) = e^{-i\theta_j Z/2}$. We denote the corresponding
unitary channel by ${\cal{Z}}(\theta_j) (\cdot): = Z(\theta_j)
(\cdot)Z(\theta_j)\hc$. Recall that the set of Clifford unitaries on
$n$ qubits forms a group generated by the single and two-qubit gates
$\{ \rm{H}, \rm{S}, \rm{CNOT} \}$. ${\cal{Z}}(\theta)_j$ decomposes as a linear
combination of Clifford unitary channels given by
${\cal{Z}}(\theta_j)(\rho) = p_{\cal{I}}(\theta_j) \rho +
p_{\cal{Z}}(\theta_j) Z\rho Z\hc + p_{\cal{S}}(\theta_j) S\rho S\hc$,
where each coefficient $p_{\cal{I}}, p_{\cal{Z}}, p_{\cal{Z}} \in
\T_{1,1}$ and is defined in Appendix
\ref{appendix:Cliff}. Alternatively, we can write
\begin{equation}
  \mathcal{Z}(\theta_j)
  = \mathcal{C}_0  + e^{i\theta_j} \mathcal{C}_{1} + e^{-i\theta_j}\mathcal{C}_{-1}
  \label{eqn:zdec}
\end{equation}
with $\mathcal{C}_0, \mathcal{C}_{\pm 1}$ linear combinations of
Clifford channels defined in Appendix \ref{appendix:Cliff} and
independent of the parameter $\theta_j$. Following our previous work~\cite{cirstoiu2017global}, we call these channels \textit{Clifford
  process modes}.  The above directly implies that the corresponding
expectation values satisfy
$\Tr( O \mathcal{Z}(\theta_j) (\rho_0)) \in \mathcal{T}_{1,1}$ for any
observable $O$ and initial state $\rho_0$.

This generalises to circuits $U(\boldsymbol{\theta}) $ as in Eq.~\eqref{eq:unitary_circuit} with independent parameters where each unitary $U(\theta_j)$ has a circuit implementation using a single parameterised one-qubit $Z(\theta_j)$ gate. For example, this occurs when each of the one-parameter unitaries $U_{j}(\theta_j) = e^{-iP_j\theta_j/2}$ are Pauli gadgets with $P_j$ an $n$-qubit Pauli operator. Then, by linearity and since the $m$ parameters are independent, it follows that given $\rho(\boldsymbol{\theta }): = U(\boldsymbol{\theta}) \rho_0  U(\boldsymbol{\theta})\hc$ the expectation value of any observable
\begin{equation}
        \Tr(O \rho(\boldsymbol{\theta }) ) \in \mathcal{T}_{m,1}
\end{equation} is a multivariate trigonometric polynomial in $\boldsymbol{\theta}$ with degree at most 1.       

If, in addition, the unitaries $V_i$ are all Clifford unitaries, and the input state $\rho_0$ is a stabiliser state, then the coefficients of the trigonometric polynomials $\Tr(\rho(\boldsymbol{\theta})P)$ where $P$ is a Pauli operator simplify to 
\begin{equation}
        c_{\bf{k} }  \in \{\pm 1/2^{||{\bf{k}}||_1}, \pm i 1/2^{||{\bf{k}}||_1}, 0\} .
\end{equation} 
where ${\bf{k}} \in \{-1,0,1\}^{m}$.
Furthermore, the expansion of $\tr(\rho(\boldsymbol{\theta}) P)$ in terms of trigonometric functions becomes particularly simple:
\begin{equation}
        \tr(\rho(\boldsymbol{\theta}) P) = \sum_{\bf{k}}  c'_{\bf{k}} \phi_{\bf{k}} (\boldsymbol{\theta})
\end{equation}
with $c'_{\bf{k}} \in \{0, \pm 1\}$ and $\phi_{\bf{k}}(\boldsymbol{\theta}) = \phi_{k_1}(\theta_1) \phi_{k_2}(\theta_2) ... \phi_{k_m}(\theta_m)$ where $\phi_{0}(\theta)=1$, $\phi_{1}(\theta) = \cos{\theta}$ and $\phi_{-1} (\theta) =\sin{\theta}$.
The proof of this result is given in Appendix~\ref{app:trig}.

\subsection{Process mode decomposition with Clifford paths}
\label{sec:proc-mode-decomp}

\label{sec:processmode}
In general situations, it may be that the circuit implementing $U(\theta_j)$ consists of several $Z$-rotations by (non-integer) real-valued multiples of the independent parameter. In such a case the resulting expected values with respect to $\boldsymbol{\theta}$ will be described by \emph{generalised m-variate trigonometric polynomials}~\cite{shestopalov2020trigonometric} of the form 
\begin{equation}
        f(\boldsymbol{\theta }) = \sum_{k_1,..., k_m = -d}^{d} c_{\bf{k}} e^{i\boldsymbol{\omega}_{\mathbf{k}} \cdot \boldsymbol{\theta}} \in \mathcal{T}_{m, d}
        \label{eq:FourierSeries}
\end{equation}
where $\mathbf{k} = (k_1,...,k_m)$ with $k_i\in \mathbb{Z}_d$ and $\boldsymbol{\omega}_{\bf{k}} = ((\omega_{\bf{k}})_1, ...., (\omega_{\bf{k}})_m )$ is a real vector of frequencies. Here, $m$ labels the number of independent parameters $\theta_i$ and each $k_i$ is an index counter of the frequencies corresponding to $\theta_i$.  The maximal number of (distinct) frequencies in a parameter is labelled by $d$. Note that generally this may take a different value for each parameter, but  we use the maximal value for simiplicity of exposition.

In order to determine the allowed values for the frequency vector $\boldsymbol{\omega}_{\bf{k}}$ and represent the corresponding coefficients $c_{\boldsymbol{k}}$ in terms of operations derived from the circuit structure we decompose the target unitary channel $\mathcal{U}(\boldsymbol{\theta})$ into a linear combination of operations separating the parameter dependency. 

\begin{definition}
A unitary channel $\mathcal{U}(\boldsymbol{\theta})= U(\boldsymbol{\theta}) (\cdot) U(\boldsymbol{\theta})\hc$ can be expressed in terms of what we call \emph{a process mode decomposition with Clifford paths} 
\begin{equation}
        \U(\boldsymbol{\theta}) =\sum_{\bf{k} =(\bf{k}_1,...,\bf{k}_m)} e^{i\boldsymbol{\omega}_{\bf{k}}\cdot \boldsymbol{\theta}} \mathcal{V}_0 \circ \mathcal{U}_{{\bf{k}}_1}\circ \V_{1}\circ ... \circ \, \mathcal{U}_{{\bf{k}}_m}\mathcal{V}_m,
\end{equation}
where each $\mathcal{U}_{{\bf{k}}_j}$ is an operation derived from the
circuit $U_{j}(\theta_j)$ by replacing each parametrised $Z$-rotation
with one of its constituent Clifford process modes, per
Eq.~\eqref{eqn:zdec}. 
\end{definition}

The length of each vector ${\bf{k}}_j$ will be given by the number of parametrised gates in the circuit, with entries taking values $\{0, \pm 1\}$ corresponding to all particular replacements. Note that the $\mathcal{U}_{{\bf{k}}_j}$ need not be completely positive trace preserving maps. The circuit implementation of $U_{j} (\theta_j)$ may also be given in terms of other multi-qubit parametrised gates, or in terms of the Hamiltonians generating it. In these cases, the process modes are constructed by generalising Eq.~\eqref{eqn:zdec}, which we describe in Appendix \ref{appendix:Cliff} in more detail. However, the frequency spectrum \[
\mathcal{S} := \{\boldsymbol{\omega_k} :  \ \forall  \, {\bf{k}}\in \mathbb{Z}_d^{m}  \}\] will be independent of any particular description of $U(\boldsymbol{\theta})$.

\subsection{Frequency spectrum for parametrised circuits}
\label{Sec:CorrelatedAngles}
The decomposition in Sec.~\ref{sec:processmode} directly gives the
trigonometric polynomial form taken by expectation values; to
emphasise their dependency on the measured observable, we employ the
notation $c_{\bf{k}}(O)$ for the coefficients.
\begin{theorem}
  \label{thm:1}
  Consider a parameterised quantum state
  $\rho(\boldsymbol{\theta})= U(\boldsymbol{\theta})\rho_0
  U(\boldsymbol{\theta})\hc$ acting on $n$ qubits with $m$ independent
  parameters
  $U(\boldsymbol{\theta}) = V_0 U(\theta_1)V_1 ... U(\theta_m)V_m$ and
  $U(\theta_j) = e^{i\theta_j H_j}$ are one-parameter unitaries
  generated by Hermitian operators $H_j$. The expectation value of any
  Hermitian operator $O$ with respect to $\rho(\boldsymbol{\theta})$
  is a generalised $m$-variate trigonometric polynomial of bounded
  degree
  \begin{equation}
    \<O\>_{\rho(\boldsymbol{\theta})}
    = \sum_{\bf{k}\in \boldsymbol \Lambda}
      c_{\bf{k}}(O) e^{i \boldsymbol{\omega}_{\bf{k}} \cdot \boldsymbol{\theta}}
    \label{eq:thm1}
  \end{equation}
  where $c_{-\bf{k}} = c_{\bf{k}}^{*}$ and the lattice
  $\boldsymbol{\Lambda} \subset \mathbb{Z}_3^{r_1}\times ... \times
  \mathbb{Z}_3^{r_m}$   where each $r_j$ is the rank of the
  Walsh-Hadamard transform of the eigenvalue vector of
  $H_j$\footnote{%
    The $2^m$-dim Walsh-Hadamard transform matrix can be
    defined as $[H_m]_{ij} := \frac{1}{2^{m/2}}(-1)^{i\cdot j}$ where
    the $i\cdot j$ indicates the bitwise dot product between the
    binary representations of the indices.
  }.  Furthermore, the frequency vector
  $\boldsymbol{\omega}_{\bf{k}} = ( (\omega_{\bf{k}})_1,\ldots,
  (\omega_{\bf{k}})_m)$ ranges over a discrete set with bounded degree
  $\sup_{\bf{k} \in \boldsymbol \Lambda}|(\omega_{\bf{k}})_j |\leq
  2||H_j||_{\infty}$.
\end{theorem}
A full proof of this result may be found in Appendix~\ref{sec:proof-thm-1} as well as its relationship to previous work that described the frequencies $(\boldsymbol{\omega_k})_j$ in terms of eigenvalue differences.

As an illustrative example, suppose
that $U_i(\theta_i)$ is given in circuit form where the only
$\theta_i$-dependent gates are  $Z(a_1\theta_i)$, ..., $Z(a_{r_i}
\theta_i)$ with fixed $a_1,..., a_{r_i}\in\mathbb{R}$. Then, the allowed
frequencies are $(\boldsymbol{\omega}_{\bf{k}})_{j}  = \mathbf{a}\cdot
{\mathbf{k}}_j$ for ${\mathbf{k}}_{j} \in \{0, \pm 1\}^{\times
  r_{i}}$ with $\mathbf{a}= (a_1,..., a_{r_i})$ and furthermore  $-
||\mathbf{a}||_1 \leq (\boldsymbol{\omega}_{\bf{k}})_{j}  \leq
||\mathbf{a}||_1$\footnote{
  For clarity of exposition we assumed $U(\boldsymbol{\theta})$ as in
  Eq.\eqref{eq:unitary_circuit}, however the result on the allowed
  frequency spectrum holds in the most general sense, where
  differently parametrised gates may be interlaced throughout the
  circuits e.g $U(\boldsymbol{\theta}) =VU_1(\theta_1) V_1
  U_2(\theta_2) V_2 U'(\theta_1)$
}.

Note that generally, the frequencies are not uniquely defined by a given path $\mathbf{k}$.
It may be the case that two different paths give $\omega_{\boldsymbol k} = \omega_{\boldsymbol k'}$, in which case the corresponding coefficient in the Fourier transform will be $c_{\boldsymbol k} + c_{\boldsymbol k'}$. Therefore in the Fourier-transform frequency spectra $\mathcal{S}$ the two paths will have indistinguishable contributions. 

This analysis involving process modes with Clifford paths may be
more suitable than the expressions in Eq.~\ref{eq:eigenvalue-expansion} to determine the allowed frequencies and coefficients in the expectation value in the following situations: when one does not have direct access to the eigenvalues;  when the parameterised unitary is given as a quantum circuit,
or when the eigenvalues have higher multiplicities where it provides a more compact form for the coefficients. Our focus is ultimately on analysing and mitigating the effects of quantum noise on expectation values of observables viewed as trigonometric polynomials, where such a quantum channel perspective is more suitable.

\section{Error diagnostics from Fourier spectra}
\label{sec:Models}

\subsection{Fourier spectra of quantum channels}
\label{sec:IIIA}

In the previous section we showed how a quantum channel approach can
be used to derive the Fourier spectra of parameterised unitary
circuits, via a process mode decomposition with Clifford paths. We now
drop the unitary requirement, and discuss Fourier spectra of general
quantum channels that may represent noisy implementations. The
previous decomposition of unitary operations into process modes can
also be seen as a consequence of the following general result.

\begin{lemma}[\cite{bennink2017unbiased}]
  \label{lemma:decompchannel}
  Any linear map $\chi$ on n qubits can be expressed as a linear
  combination of Clifford unitary channels $C_i$ and Pauli reset
  channels\footnote{%
    A \emph{Pauli reset channel} maps any input state to a
    specific fixed eigenstate of a Pauli operator.
  } $S_j$ 
  \begin{equation}
    \chi = \sum q_i C_i + \sum q'_j S_j
  \end{equation}
  for $q_i , q_i'\in \mathbb{R}$. Trace preservation is equivalent to
  $\sum_i q_i + \sum_j q_j'= 1$, and the $S_j$ terms are only required
  for non-unital channels\footnote{%
    Note that there is no simple rule to determine if the resulting
    linear map will be completely positive, and therefore physical.
  }.  The decomposition is not unique. 
\end{lemma}

In general, the Fourier spectrum of an expectation value evaluated on
a parameterised state produced by a noisy channel will differ from
that of a noiseless one in two distinct ways: the Fourier coefficients
might be different, and the frequency support might vary.

\begin{theorem}
  \label{thm:noisy-expectations}
  Writing
  $\tilde{\rho}(\boldsymbol{\theta}) =
  \chi(\rho(\boldsymbol{\theta}))$ for the noisy implementation of the
  target state $\rho(\boldsymbol{\theta})$, and allowing the noise channel
  $\chi$ to depend on $\boldsymbol{\theta}$, we have
  \begin{equation}
    \<O\>_{\tilde{\rho} (\boldsymbol{\theta})}
    = \sum_{\bf{k}\in \bf{\Lambda}} \sum_{j}
      q_{j}(\boldsymbol{\theta}) c_{\bf{k}}(\mathcal{C}_j\hc(O))
      e^{i\boldsymbol{\omega}_{\bf{k}}\cdot \boldsymbol{\theta}}
    + q'(\boldsymbol{\theta}) 
    \label{eqn:generalnoise}
  \end{equation}  
 where $q'(\boldsymbol{\theta} ):= \sum_j q_{j}'(\boldsymbol{\theta})
 \<P_j|O|P_j\>$ for some $|P_j\>$ eigenstates of Pauli operators in
 the decomposition of Pauli reset channels.
\end{theorem}
\noindent
See Appendix~\ref{sec:proof-thm-2} for the proof and additional details.

Recall that in the noiseless case, a discrete set of frequencies
suffices, corresponding to the points of the lattice
$\Lambda \subset \mathbb{R}^{m}$.  In the noisy case, the expectation
value $\<O\>_{\tilde{\rho} (\boldsymbol{\theta})}$ is still a function
of $\boldsymbol{\theta}$, but it may no longer have a discrete Fourier
series.  Therefore we define the \emph{noisy frequency spectrum} as
\[
  \tilde{\mathcal{S}}:= \{
    \tilde{\boldsymbol{\omega}} \in \mathbb{R}^m :
    \mathcal{F}[\<O\>_{\tilde{\rho}}](\tilde{\boldsymbol{\omega}})
    \neq 0
  \}
\]
where $\mathcal{F}$ denotes the Fourier transform, computed as
\[
  \mathcal{F}[\<O\>_{\tilde{\rho}}](\tilde{\boldsymbol{\omega}})
  =
    \int \<O\>_{\tilde{\rho}(\boldsymbol{\theta})}
    e^{i \tilde{\boldsymbol{\omega}}\cdot \boldsymbol{\theta}}
    d\,\boldsymbol{\theta} .
\]

The difference between $\mathcal{S}$ and $\tilde{\mathcal{S}}$ depends
on the specific noise parameters $q_j(\boldsymbol{\theta})$, and
different channels have distinct effects on the spectra.  In the
following section we illustrate this for several noise models based on
varying physical assumptions.

\subsection{Effect of decoherent channels}
\label{sec:effect_decoh}

Consider a noisy implementation where each unitary gate operation is
followed by a noise channel that is independent of the parameters
\begin{equation}
  \label{eq:simplenoisychannel}
  \tilde{\U}
  = \tilde{\V}_{0} \circ \U_{1}(\theta_1) \circ \tilde{\V}_{1} \circ \cdots \circ \U_{m}(\theta_m) \circ \tilde{\V}_m.
\end{equation}
It is clear that for any fixed observable  $\tilde{\S} \subseteq \S$,
in other words the range of frequencies remains the same. However the
magnitude of the Fourier coefficients might be
different. Stronger assertions are possible if we specialise to
particular cases.

For example, if each noisy gate $\V_i$ is affected by a (global)
depolarising channel with parameter $p$, then one finds that the noisy
Fourier coefficients $\tilde{c}_{\bf{k}}$ are homogeneously contracted
\begin{align}
  \<O\>_{\tilde{\rho}(\boldsymbol{\theta})}
  = &\sum_{\bf{k} \in \boldsymbol \Lambda} (1-p)^{m+1} c_{\bf{k}}(O)
      e^{i\omega_{\bf{k}}\cdot \boldsymbol{\theta}} \nonumber \\
    &\quad\quad+  \frac{(1-(1-p)^{m+1})}{2^n}\Tr(O).
\end{align}
This can be seen as a special case of the following result.

\begin{lemma}[Pauli channel]
  Let $\tilde\U$ be as in Eq.~\eqref{eq:simplenoisychannel} and suppose
  each $\tilde{\V}_i$ is a Clifford unitary followed by a Pauli
  channel\footnote{%
    A \emph{Pauli channel} is an incoherent convex combination of
    Pauli errors.
  }, and each $\U_{i}$ is a Pauli gadget; then, for any Pauli
  observable $O$, the coefficients are contracted for all paths so
  that $|\tilde{c}_{\bf{k}}(O)| \le |c_{\bf{k}}(O)|$.
\end{lemma}

For general observables $O$ and non-Clifford $\V_i$, some modes
that were previously zero due to cancellation of paths might find
themselves with a non-zero value under a Pauli noise model.

In Appendix~\ref{app:noisesFourier}, we provide a comprehensive analysis
of noise channels and their effect on Fourier
coefficients, summarised in Table~\ref{tab:noisesFourier}. 
Importantly, the effect of coherent, non-unital and parameter-dependent channels is detectable from the Fourier spectrum as additional modes. Parameter-dependent channel in particular may lead to modes with frequencies outside the theoretical spectrum $\S$.

\begin{table*}[htbp!]
\begin{tabular}{|l|l|c|}
\hline
Noise channel& Effect on Fourier coefficients & $\tilde\S \subseteq \S$?\\
\hline
\hline
Depolarising (Lemma \ref{lem:dep}) & Uniform contraction&  \cmark \\
\hline
Pauli, SPAM (Lemma \ref{lem:pauli}) & Contraction ($O$ Pauli and $\V$ Clifford)& \cmark \\
\hline
Aligned Pauli (Lemma \ref{lem:aligned}) & Uniform contraction (any $O$ and $\V$)& \cmark \\
\hline
Coherent (corollary of Lemma \ref{lem:purity}) & Different coefficients but conserved norm of coefficient vector $\vec{\tilde c}$& \cmark \\
\hline
Non-unital (Lemma \ref{lem:nonunital}) & Different coefficients, new modes with subset of parameters& \cmark \\
\hline
Parameter-dependent (Lemma \ref{lem:corr}) & Different coefficients, new modes outside theoretical spectrum& \xmark \\
\hline
\end{tabular}
\caption{Effect of different noise channels on Fourier coefficients of
  the landscape. Proofs can be found in
  Appendix~\ref{app:noisesFourier}. Rightmost column indicates whether
  the noisy spectrum $\tilde\S$ is included in the theoretical
  noiseless spectrum $\S$.}
  \label{tab:noisesFourier}
\end{table*}

\subsection{Shot noise}
\label{sec:finite-size}
Overall the analysis in previous sections confirms that noise is capable of expanding the frequency spectrum. With knowledge of the exact, noiseless spectrum, one thus has a direct way of quantifying the impact of noise on a quantum variational landscape. In practical terms, an implementation of error diagnosis from Fourier data would involve performing a discrete Fourier transform (DFT) by sampling parameter values. 
Since the number of points determines the resolution at which the frequencies in the noisy spectrum can be determined, the grid should be fine enough to be able to resolve frequencies beyond $\mathcal{S}$. 
In Appendix~\ref{app:resource} we discuss possible approaches to optimising the number of calls to the quantum computer.

These sampling requirements are in turn different to those needed to resolve the individual points of the landscape with sufficient accuracy. We consider this requirement in the next section.

Typically, for DFT the sampled points $\boldsymbol{\theta}_{\boldsymbol{i}}$ are equally spaced on a grid $\boldsymbol{\Theta} = \frac{2\pi}{d} \mathbb{Z}_d^{\times m}$, where $d = 2N + 1$ for $N$ the largest frequency resolvable. For each sampled point, the corresponding circuit $U(\boldsymbol{\theta}_i)$ is repeated $n_{s}$ times to construct sample mean estimators $\tilde{x}_{i}$ for the expectation value $x_i$. From the central limit theorem, the standard deviation due to a finite number of shots scales as $ 1/\sqrt{n_s}$ and using properties of the DFT a similar scaling occurs in estimating the Fourier coefficients under shot noise.

\begin{lemma}[Shot noise]
Given a number of shots $n_s$, each Fourier coefficient is normally distributed around the noiseless mean with standard deviation
\begin{equation}
        \sigma(\Re\,(\mathcal{F}[\tilde{x}]_{\boldsymbol{k}}))= \frac{\sqrt{1 - ||x||_2^2/d^m}}{\sqrt{2 n_s d^m }}
\end{equation}
and similarly for the imaginary part.
\end{lemma}
Two observations can be made on this result: first, the standard deviation of each Fourier coefficient due to shot noise is independent on the coefficient's frequency vector $\boldsymbol k$, and second, the dependence on $d^m$ means that increasing the resolution in the frequency spectrum or the number of parameters varied reduces the shot noise in the Fourier coefficients with a scaling of $1/O(\sqrt{d^m})$.

\subsection{Figures of merit from Fourier spectrum}
\label{sec:figures_of_merit}

\subsubsection{Power spectra}
\label{sec:power-spectra}

If a frequency spectrum is known to be included into  $\S$, one can quantify the extra noise-induced modes from an experimentally measured spectrum $\tilde{\S}$ using the signal power (or energy).
\begin{definition}[Quality measures]
 The \textbf{power in the extra modes} due to noise is given by
\begin{equation}
        P_N := \sum_{\boldsymbol \omega \in \tilde{\S}/\S} |\tilde c_{\boldsymbol \omega}|^2
\end{equation}
while the \textbf{signal power} will be given by
\begin{equation}
        P_S := \sum_{\boldsymbol \omega \in \S} |\tilde c_{\boldsymbol \omega}|^2
\end{equation}
The \textbf{total power} $P := P_N + P_S$ is proportional, via Parseval's theorem, to the power of the real-space signal $\sum_i |x_i|^2$. 
One can also define the observed \textbf{signal-to-noise ratio} (SNR) as $SNR := P_S/P_N$.
\end{definition}

Distinguishing noise from the signal will depend on the knowledge of the spectrum $\S$, as well as the size of the experimentally measured spectrum $\tilde{\S}$. To maximise the distinction, it is beneficial to identify $\S$ as small as possible, which might be possible for specific circuit structures. In Section \ref{sec:mitigation} we show how the noiseless spectrum can be refined for the case of QAOA. However, for general circuits one is confined to the results of Theorem \ref{thm:1}.

\subsubsection{Average fidelity and purity}
\label{sec:aver-fidel-purity}

A quantity of practical relevance for benchmarking that can be extracted from the Fourier representation is the average fidelity over all sampled parameter values $\langle F \rangle_{\boldsymbol \theta}$, equal to $\mathcal{F}[F]_{\boldsymbol 0}$. Let us represent the Fourier coefficients over all Pauli operators in the vectorised notation $\vec c$, where $[\vec c]_{\boldsymbol k, i} = c_{\boldsymbol k}(P_i)$. Therefore one can show that 
\begin{lemma}[Average fidelity]
The average fidelity of the output state over all sampled parameters is proportional to the inner product of the vectors of Fourier coefficients:
\begin{equation}
        \langle F\rangle_{\boldsymbol\theta} = \int F(\tilde{\rho}(\boldsymbol{\theta}), \rho(\boldsymbol{\theta})) d\, \boldsymbol{\theta}= \frac{1}{2^n} \vec c\,^\dagger \vec{\tilde c}
\end{equation}
where $ F(\tilde{\rho}(\boldsymbol{\theta}), \rho(\boldsymbol{\theta}))$ is the fidelity between the target state $\rho$ and the noisy implementation $\tilde{\rho}$.
\end{lemma}
From this it is clear that the Fourier modes in $\vec{\tilde c}$ that fall outside the support of $\vec c$ do not contribute to the average fidelity. In other words, additional modes such as those produced by parameter-dependent errors are not captured by this averaged measure, but nonetheless will affect fidelity of individual circuits.

Of course, in a practical setting it is unrealistic to measure all Pauli observables, while direct fidelity estimation \cite{flammia2011direct} requires at worst $\Omega(2^n)$ measurements, for many classes of states resource requirements are reduced, and one could also further employ estimation of multiple observables using classical shadows to reduce the computational burden \cite{huang2020predicting}.

Similarly, the average purity can be related to the norm of the noisy Fourier coefficient vector:
\begin{lemma}[Average purity] \label{lem:purity}
The average purity of the output state $\tilde\rho(\boldsymbol\theta)$ over all sampled parameters is proportional to the squared norm of the vector of Fourier coefficients
\begin{equation}
        \langle P\rangle_{\boldsymbol\theta} = \frac{1}{2^n} |\vec{ \tilde{c}}|^2.
\end{equation}
\end{lemma}
The relation between purity and spectral power is therefore that the average power over all Pauli operators is proportional to the average purity of the state over all sampled parameters.
As a corollary, coherent noise channels in the preparation circuit preserve purity of the final state and therefore the norm of $\vec{\tilde{c}}$.

Furthermore, since coherent errors preserve purity of the state, it must be that under coherent noise channels the norm of $\vec c$ is preserved. Vice versa, decoherent channels will decrease purity, and therefore lead to a contraction in the components of $\vec c$, consistent with the analysis presented in Sec.~\ref{app:noisesFourier}.

\subsection{Experiments}
\label{sec:experiments}

\subsubsection{Simple circuit}
\label{sec:simple-circuit}

In the following we apply spectral noise evaluation to example circuits run on real quantum devices.
The first circuit examined is a two-qubit circuit, consisting of two parameterised $R_y$ rotations, and the measurement operator is the parity operator $ZZ$. The parameterisation is the canonical one, therefore the noiseless expectation value is $f(\theta_1, \theta_2) = \cos(\theta_1)\cos(\theta_2)$, which corresponds to 4 nonzero Fourier coefficients with maximum absolute frequency 1. Therefore, to study the presence of higher frequency modes we sample up to a maximum absolute frequency of 2, giving a 5x5 grid. The experiment was run on a superconducting device (\texttt{ibmq\_lima}).

To assess the effect of noise on the frequency spectrum, after state preparation and prior to measurement we introduce a series of identity gates, composed of two successive CNOT gates. These would cause an increase in noise by extending the depth and therefore computation time, increasing decoherence, and by being noisy themselves.
Furthermore, we note that error mitigation via extrapolation for digitised computations~\cite{giurgica2020digital} involves introducing such fiducial gates to artificially boost noise levels, and so the analysis here may also be used to test noise assumptions involved in deriving fitting functions.

In Figure~\ref{fig:power_simple_circuit} we show the total power in noise modes, defined above as modes not present in the noiseless spectrum. The latter is further divided in noise on the expected noiseless spectrum $\S$ (`on support') and noise outside of it (`off support' ), which in this case consists of modes having any components with frequency higher than 1.
The error bars correspond to one standard deviation as evaluated by a bootstrapping variance estimation method~\cite{young1996jackknife}.
\begin{figure}[ht]
        \includegraphics[width=0.5\textwidth]{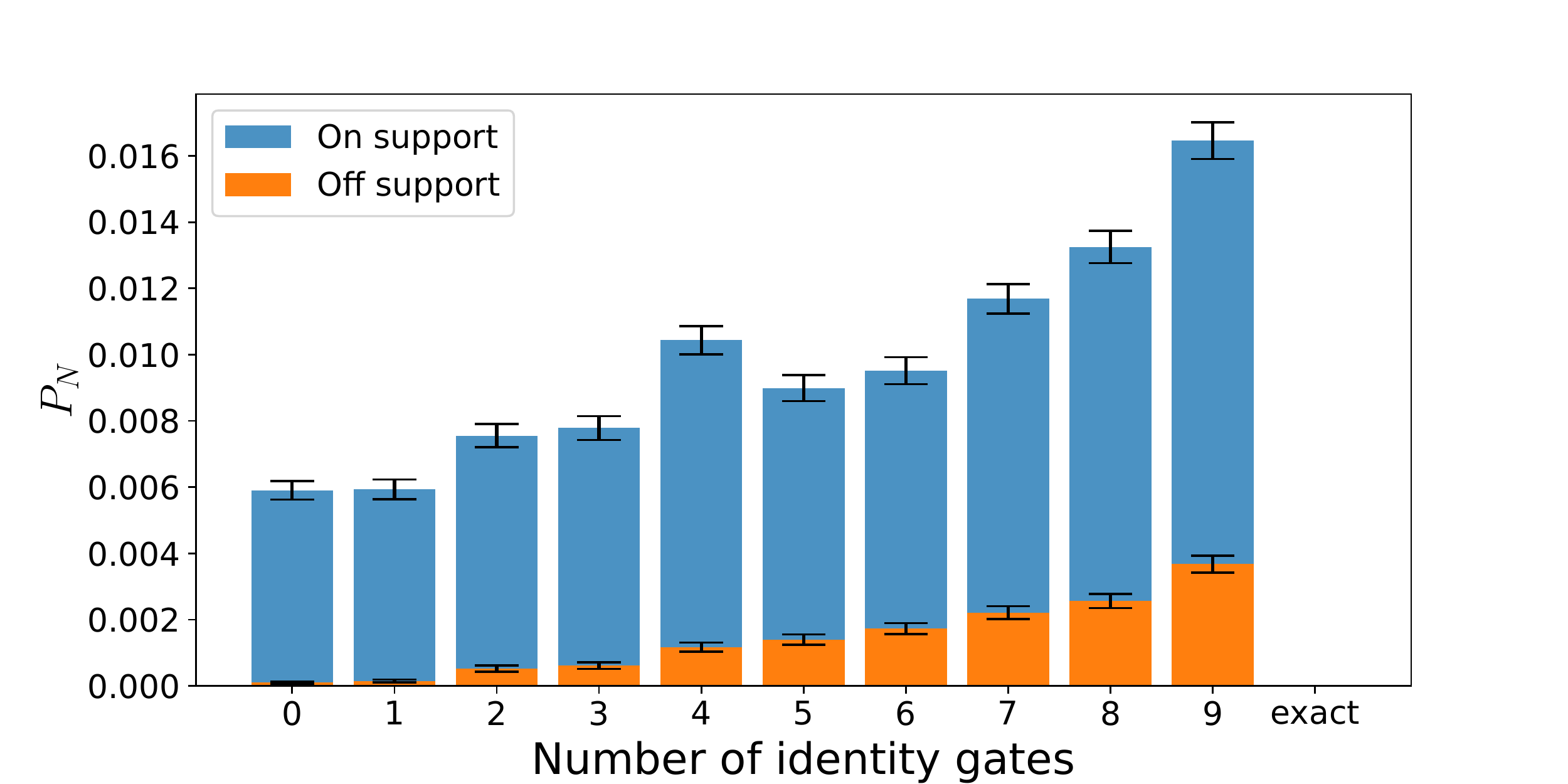}
        \caption{Power in noise modes for signal measuring the observable $O=ZZ$ for 2-qubit 2-rotations circuit, for increasing depth of fiducial gates corresponding to a logical identity (in this case two CNOT operations). The experiment was run on \texttt{ibmq\_lima}.}
        \label{fig:power_simple_circuit}
\end{figure}
The results reveal that, as expected, increasing the number of identity gates has a progressive detrimental effect on the quality of the output. Most of the noise appears to come from on-support modes, while off-support extra modes remain minor, however they also appear to increase progressively with depth. Based on the previous discussion, the decay in signal power is to be interpreted as likely being caused by decoherence, while the increase in additional modes, both on- and off-support, signals the increased effect of more complex noise channels. 

\subsubsection{UCC circuit}
\label{sec:ucc-circuit}

Next, we evaluate the spectrum of a more complex 2-qubit circuit inspired by the unitary coupled cluster (UCC) approach to VQE~\cite{mcardle2020quantum}, consisting of two Pauli gadgets with $P_1=XY$ and $P_2=YX$. The measurement operator is $ZI$, which yields the expectation value $f(\theta_1, \theta_2) = \cos(\theta_1)\cos(\theta_2) - \sin(\theta_1)\sin(\theta_2)$, giving 2 Fourier coefficients. This time, the experiment is run on four superconducting devices provided by IBM, as well as the ion trap computer Quantinuum H1-1, powered by Honeywell. The sampling grid is identical to the previous experiment. The results are shown in Figure~\ref{fig:merit_ucc}.
\begin{figure}[htp]
        \includegraphics[width=0.48\textwidth]{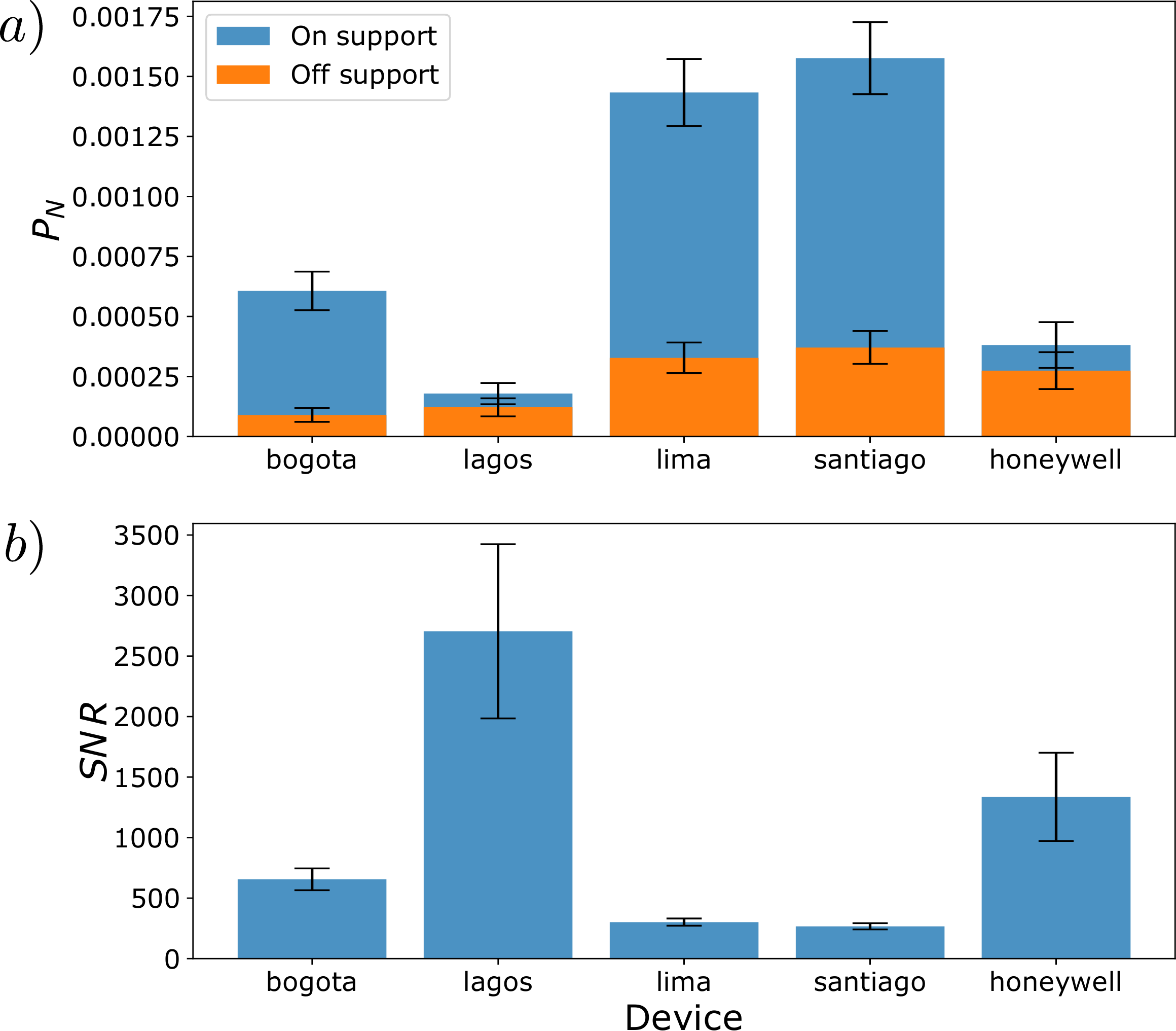}
        \caption{Figures of merit for 2-qubit UCC circuit, a) power in noise modes, b) signal-to-noise ratio. Measured observable $O=ZI$. Run on different systems: Quantinuum H1-1 (ion-based), IBM's \texttt{ibmq\_santiago}, \texttt{ibmq\_lima}, \texttt{ibmq\_bogota}, \texttt{ibmq\_lagos} (superconducting).}
        \label{fig:merit_ucc}
\end{figure}

Overall, the devices with the worst performance are \texttt{ibmq\_santiago} and \texttt{ibmq\_lima}, while the best performing on the task are HQ1 and \texttt{ibmq\_lagos}. This might be ascribed to their greater number of qubits (10 and 7), giving a larger parameter space for the compiler to optimise noise, as well as to higher quantum volume in the case of H1-1 (which, at the time of writing, holds the world record for this measure~\cite{honeywell2021qv}). On the other hand, the other machines have all a lower number of qubits (5).
Interestingly, for the most accurate devices the majority of the noise modes fall outside the exact support, hinting that these devices might benefit the most from spectral analysis and noise mitigation.

Therefore we see that in principle spectral methods could be valuable in benchmarking quantum devices by providing multiple figures of merit, and separating different physical noise effects. However, we stress that the experiments presented here are extremely limited in both system size, circuit depth and number of measurement operators sampled. The results cannot be considered as anything beyond a rough assessment of the performance of the different machines in this particular task, and are meant to serve as a proof-of-principle rather than a meaningful comparison.

\section{Error mitigation based on spectral information}
\label{sec:error-mitig-based}

In Appendix~\ref{sec:AppendixFilter} we prove that below a noise threshold, a signal can be perfectly recovered by rounding its Fourier coefficients for a class of quantum circuits. These contain parametrised gates with independent parameters and unparameterised gates that are Clifford but otherwise unknown. However, this theorem does not extend to general circuits, where typically one has both unparameterised non-Clifford gates and dependent parameters. While the result applies to circuits of practical interest such as the unitary coupled clustered ansatz \cite{romero2018strategies,cowtan2020generic}, the classical complexity of the fast Fourier transform scales exponentially in the number of parameters thus limiting the theorem's use to circuits containing few non-Clifford gates. On the other hand, circuits with a number of non-Clifford gates scaling logarithmically in the number of qubits can be efficiently simulated \cite{bravyi2016improved}. 

However, given knowledge of the frequency spectrum for a circuit family, we design  application-specific error mitigation strategies that allow to approximately reconstruct the signal. In Appendix~\ref{sec:AppendixFilter} we outline two such strategies inspired by well-known classical spectral filtering methods, namely thresholding and filtering. Both aim to recover the clean signal by isolating the most significant components from a noisy spectrum. Thresholding assumes that the noise modes are small in magnitude, and therefore sets a lower bound to the coefficient's absolute value, while filtering uses knowledge of the noiseless spectrum to remove all frequencies that must be induced by noise.

In the next sections we combine these strategies with a decoherence-mitigating method and demonstrate their usefulness in a realistic scenario.

\subsection{Spectral denoising with near-Clifford circuits}
\label{sec:denoising_near_cliff}
A general feature of incoherent quantum noise is the scaling of signal coefficients, as explained in Section~\ref{sec:IIIA}.  Denoising techniques based on frequency filtering or thresholding do not change the value of the significant coefficients, and therefore would not reduce the effect of decoherent processes on their own. This can be viewed as an indication that Fourier methods should be complemented with other noise mitigation techniques designed to reduce decoherence effects.
Here, we explore a combination of Fourier methods with Clifford Data Regression (CDR)~\cite{czarnik2020error}, which has shown promise as a practical mitigation method for expectation values. CDR has features that make it natural to integrate with Fourier methods, as it relies on learning noise based on the entanglement structure of the circuit, independently from the angles of rotation. The combination thus leads to a more efficient use of samples.

Denote by $\tilde{f}(\boldsymbol{\theta})$ the noisy expectation value of a given observable $O$ with respect to the parameterised circuit $U(\boldsymbol{\theta}) $ acting on input state $\rho_0$
and by $f^{\#}(\boldsymbol{\theta})$ its restriction to the frequency spectrum $\boldsymbol{\omega}\in\mathcal{S}$ for the input circuit as determined using the techniques illustrated in Appendix~\ref{sec:AppendixFilter}.
In practice, when implementing these protcols on a quantum device, one deals with sample mean estimators of the noisy expectation value $\tilde{f}$ which, along with implementing a discrete Fourier transforms gives an estimator for $f^{\#}$ that includes shot noise as discussed in Sec.~\ref{sec:finite-size}. 
After filtering noise-induced modes, the second step would be to reduce the effect of incoherent noise by determining the noise scaling $A_{\boldsymbol{\omega}}$ that give the error-mitigated estimators $\hat{c}_{\boldsymbol{\omega}}:= A_{\boldsymbol{\omega}} \tilde{c}_{\omega}$ of the exact Fourier coefficients $c_{\boldsymbol{\omega}}$. To this aim, there are several strategies one can employ. In our experiments we have explored the following strategies

The simplest approach is to first apply CDR, as it was introduced in \cite{czarnik2020error}, in order to learn scaling factors $A$ and $A'$ that minimise 
 \begin{equation}
        A,A' : = arg \min_{a,a'}\frac{1}{|\mathcal{C}|}\sum_{C\in \mathcal{C}}| f_{C} - a \tilde{f}_C -a'|
 \end{equation}
where $\tilde{f}_{C}$ and $f_{C}$ denote respectively the noisy and the classically simulated expectation value of the observable $O$ with respect to state $C\rho_0 C\hc$. $C$ is a near-Clifford circuit derived from the target parameterised circuit $U(\boldsymbol{\theta})$ where all but a fixed number of $D$ non-Clifford gates are replaced with Clifford operations. Once $A$ and $A'$ have been found, we filter out noise-induced modes using the methods in Sec.~\ref{sec:IIIA} and~\ref{sec:IIIB} to give an error-mitigated estimator.
\begin{equation}
 \label{eq:cdr} 
        \hat{f}(\boldsymbol{\theta}) = Af^{\#} (\boldsymbol{\theta}) + A'
\end{equation}
for the target expectation value $f(\boldsymbol{\theta})$. To further reduce quantum compute resources, we include (part of) the circuit evaluations required for the Fourier transform into the training set $\mathcal{C}$.

Alternatively, one may implement CDR by training the model on spectrally filtered data. There is the limitation that the training data must now be in a form conducive to Fourier mitigation, and thus a minimum of non-Clifford gates must be present. This issue is exacerbated whenever, for a given parameter $\theta$, multiple $\theta$-dependent gates occur in the synthesised target circuit, as they cannot be replaced by independent Clifford gates in the training phase.
Lastly, a more involved strategy rests on the observation that not all coefficients corresponding to $\boldsymbol{\omega}\neq 0$ are scaled by the same amount under the action of real hardware noise. This behaviour may be seen in the raw Fourier spectrum obtained on IBMQ devices (Figures~\ref{fig:qaoa_1p} and \ref{fig:qaoa_2p}) but also in the analysis of different noise models in Sec.~\ref{sec:Models}.
Using a similar training set as in the second approach above, one then would seek to learn a vector ${\bf{A}} = (A_{\boldsymbol{\omega}})$ of parameters that minimise
  \begin{equation}
        {\bf{A}} : = {\rm{arg}} \min_{a_{\boldsymbol{\omega}}}\frac{1}{|\mathcal{C}|}\sum_{C\in \mathcal{C}}| f_{C} - \sum_{\boldsymbol{\omega}\in\mathcal{S}} a_{\boldsymbol{\omega}}\tilde{c}_{\boldsymbol{\omega}} e^{i\boldsymbol{\omega}\cdot \boldsymbol{\theta}_C}|,
 \end{equation}
where $\boldsymbol{\theta}_C$ corresponds to the near-Clifford circuit $C$, with each $(\theta_{i})_C$ parameter chosen so that the corresponding $\theta_i$ parameterised gates are Clifford.
This suffers from analogous problems in constructing the training set. 

In our experiments all three methods were attempted, however we found that more involved methods did not amount to a significant improvement in the mitigation results, while significantly increasing the computational burden. Therefore, we limit our analysis to the simplest method.

\subsection{Experiments}
\label{sec:mitigation}
To demonstrate the applicability of noise filtering methods, we give proof-of-principle implementations on current quantum hardware for circuit structures used in near-term algorithms.

The magnitude threshold for determining the most significant Fourier coefficients is calculated as 
\begin{equation}
T = B\,\underset{\boldsymbol\omega \in \tilde\S}{\text{median}}\,(|\tilde c_{\boldsymbol\omega}|)
\end{equation}
for some tunable constant $B$. Assuming the majority of
coefficients in the sampled frequency spectrum are due to noise (which in principle requires the resolved spectrum to be double the size of the noiseless one), and further that these are small and of similar magnitude, this allows us to select only the outliers and hence recover the signal coefficients. The CDR is performed with a simple linear model as in Equation~\ref{eq:cdr}. The training set $\C$ is generated by substituting for each target circuit the closest Clifford to every non-Clifford gate, except for $D$ randomly chosen gates.

As noted before, spectral noise filtering is expected to yield a landscape that has more similar features to the noiseless one, while a scaling technique like CDR is necessary to reduce the error in the magnitudes. Therefore, we focus on two quantities of interest that capture these two different scenarios.

For each noise mitigation experiment we report the similarity measure between the mitigated vector $C^\#$ and the noiseless vector $C$ obtained from a unitary simulator (Qiskit Aer's statevector simulator \cite{qiskit2021statevector}). 
The simplest measure of similarity between vectors of data points is the Euclidean distance $||C^\# - C||$, which has a practical interpretation as the root mean squared error between the data points. Furthermore, it is invariant (up to a constant) under the Fourier transform, due to Parseval's theorem, and therefore may be easily calculated from the spectrum.
We then choose to consider the additional measure of cosine similarity
\begin{equation}
        S := \frac{C^{\#\dagger}C}{||C^\#||\,||C||}
\end{equation}
This quantity is robust to global scaling of any one of the two vectors. The quantity can therefore be visualised as a scale-invariant measure of error: indeed, if both vectors are unit-normalised, cosine similarity is monotonically related to Euclidean distance: $||\frac{C^\# }{||C^\#||} - \frac{C}{||C||}|| = \sqrt{2 - 2S}$.

The experiments are performed on IBMQ superconducting devices.

\subsubsection{UCCSD}
\label{sec:uccsd}

\begin{table*}[htb!]
\centering
\begin{tabular}{ |p{3cm}||p{1cm}|p{1cm}|p{1cm}|p{1cm}||p{1cm}|p{1cm}|p{1cm}|p{1cm}| }
 \hline
 \multicolumn{1}{|c|}{}&\multicolumn{4}{|c|}{Cosine similarity}&\multicolumn{4}{|c|}{Euclidean distance} \\
 \hline
 Methods& A& B& C& D& A& B& C& D \\
 \hline
 \hline
 Raw&                                   0.824&                  0.664&                  0.868&                  0.881&                  0.448&                  0.539&                  0.448&                  0.542           \\
 \hline
 CDR&                                   0.823&                  0.665&                  0.864&                  0.873&                  0.319&                  0.462&                  0.331&                  0.396           \\
 \hline
 Filtering + CDR&               0.835&                  \bf{0.675}&     \bf{0.886}&     \bf{0.912}&     0.316&                  \bf{0.461}&     \bf{0.327}&     \bf{0.391}      \\
 \hline
 Thresholding + CDR&    \bf{0.841}&             0.672&                  0.871&                  0.909&                  \bf{0.314}&     0.462&                  0.329&                  \bf{0.391}      \\
 \hline
\end{tabular}
\caption{Effectiveness of landscape noise mitigation methods on 4-qubit UCC circuit for $H_2$. Singles excitations parameters are varied, doubles excitations are fixed to random values. A: \texttt{ibmq\_manila}, $O=ZIZI$; B: \texttt{ibmq\_manila}, $O=XXYY$; C: \texttt{ibmq\_lima}, $O=ZIZI$; D: \texttt{ibmq\_lima}, $O=XXYY$. Note that higher cosine similarity and lower Euclidean distance are better.}
\label{tab:ucc}
\end{table*}

We begin by analysing a UCCSD circuit for the H2 molecule under the STO-3g basis set, which requires four qubits under Jordan-Wigner representation. The circuit has two parameters corresponding to `singles' excitations, each appearing on two Pauli gadgets, and one parameter governing `doubles' excitations, composed of eight further gadgets. The circuit has a total CNOT depth of 52 and no depth-reduction compilation passes were applied. We choose to focus on the two singles parameters for spectral noise mitigation. Since they each control two single-qubit rotation gates, the noiseless support has a maximum frequency of 2, therefore we sample on a 7x7 grid giving a maximum frequency resolution of 3.

In the target circuit, the Pauli gadgets implementing the doubles excitations are chosen independently and set to fixed angles normally distributed around $\pi/4$ with standard deviation $0.1$. By randomising the fixed parameters instead of setting them to the optimum we aim to illustrate an improved profile of the noisy energy landscape for general parameter values, allowing to evaluate the mitigation method in a realistic setting where the optimal parameters are not known.
For CDR we set $D = 2$ and for thresholding constant we find that $B = 3$ is sufficient.

The results are shown in Table~\ref{tab:ucc}. Both cosine similarity and Euclidean distance are consistently improved by the combination of spectral noise mitigation and CDR compared to no mitigation and CDR alone. The improvement due to spectral techniques is especially notable for cosine similarity, while the decrease in Euclidean distance is more modest. This reflects the fact that Euclidean distance is highly dependent on relative scaling between the vectors, which is what CDR targets by construction.

\subsubsection{QAOA}
\label{sec:qaoa}

\begin{table*}[htb!]
\centering
\begin{tabular}{ |p{3cm}||p{1cm}|p{1cm}|p{1cm}|p{1cm}|p{1cm}||p{1cm}|p{1cm}|p{1cm}|p{1cm}|p{1cm}| }
 \hline
 \multicolumn{1}{|c|}{}&\multicolumn{5}{|c|}{Cosine similarity}&\multicolumn{5}{|c|}{Euclidean distance} \\
 \hline
 Methods& A& B& C& D& E& A& B& C& D& E\\
 \hline
 \hline
 Raw&                                   0.838&                  0.764&                  0.930&                  \bf{0.764}&             0.789&                          1.791&                  \bf{2.526}&             1.803&                  \bf{2.416}&     3.854\\
 \hline
 CDR&                                   0.863&                  0.765&                  0.934&                  0.751&                  0.827&                          1.376&                  2.581&                  1.468&                  2.811&                  3.229\\
 \hline
 Filtering + CDR&               0.868&                  \bf{0.766}&             0.940&                  0.752&                  0.833&                          1.371&                  2.581&                  1.465&                  2.811&                  3.225\\
 \hline
 Thresholding + CDR&    \bf{0.903}&             0.762&                  \bf{0.979}&             0.753&                  \bf{0.878}&                     \bf{1.345}&             2.585&                  \bf{1.451}&             2.811&                  \bf{3.199}\\
 \hline
\end{tabular}
\caption{Effectiveness of landscape noise mitigation methods on MaxCut QAOA circuit for random 3-reg graph. The parameters varied are the last two $\beta,\gamma$. Where $p=2$, the parameters n the first layer are set to $\pi/4$. The measurement operator is $O=H$. A: \texttt{ibmq\_guadalupe}, $p=1$, 8 qubits; B: \texttt{ibmq\_guadalupe}, $p=2$, 8 qubits; C: \texttt{ibmq\_sydney}, $p=1$, 8 qubits; D: \texttt{ibmq\_sydney}, $p=2$, 8 qubits; E: \texttt{ibmq\_guadalupe}, $p=1$, 16 qubits. Note that higher cosine similarity and lower Euclidean distance are better.}
\label{tab:qaoa}
\end{table*}

We then consider a quantum approximate optimization algorithm (QAOA) circuit \cite{farhi2014quantum} for the MaxCut problem on random 3-regular triangle-free graphs of various sizes, with $V$ and $E$ denoting the vertices and edges. The circuits involve repeated layers labelled by $p$, whereas each layer alternates between a unitary generated by the cost Hamiltonian $H_C = \sum_{(i,j)\in E} Z_{i}Z_{j}$ and a mixing Hamiltonian $H_D = \sum_{i\in V}  X_i$,  with the initial state an equal superposition over the computational basis. The $p=1$ layer circuit on $n$ qubits produces the parametrised state $e^{-i\beta H_D}e^{-i\gamma H_C}|+\>^{\otimes n}$. The benefits of this class of QAOA problems for benchmarking noise mitigation algorithms are twofold: first, the scaling in depth with system size is favourable compared to UCC, as long as the number of layers is kept small. Second, since the Hamiltonian is diagonal in the Z basis the cost function value can be calculated from the single shots without adding additional rotation gates. 
We compare the same noise mitigation techniques as UCC, on an 8-qubit QAOA circuit with $p=1,2$ on both the 16-qubit \texttt{ibmq\_guadalupe} and the 27-qubit \texttt{ibmq\_sydney} quantum computers. We also run a 16-qubit $p=1$ experiment on \texttt{ibmq\_guadalupe} as an example of a larger system. As target function we measure the expectation value of the Hamiltonian.

All experiments are run on two parameters. In the $p=2$ experiments we choose to focus on the parameters in the last layer, in which case we set both parameters of the first layer to $\pi/4$.
By using a lightcone argument one can then show that in all experiments that the frequencies will be limited to 2 for the mixing parameter $\beta$, and 5 for the problem parameter $\gamma$, which follows from the problem graphs being all 3-regular and triangle-free. Therefore, we sample on a 13x13 grid to ensure a resolution of 6 in frequency space. 

Training CDR on this system presents with the additional complication that setting all parameters to give Clifford gates leads to an expectation value of zero. Therefore, in order to obtain training data of sufficient magnitudes the CDR parameter $D$ is increased from 2 to 10, and the data is post-selected to only train on the largest magnitude datapoints. Due to the higher noise level we find that increasing the thresholding constant to $B = 10$ is more appropriate.

The results are shown in Table~\ref{tab:qaoa}. The performance of noise mitigation techniques is seen as highly dependent on the depth of the circuit. Circuits with $p=1$ show consistent improvements with noise mitigation and especially following spectral filtering. Interestingly, thresholding appears to be the most effective filtering method, which can be justified by the noiseless signal being sparse and much higher than the noise level, as shown in Figure~\ref{fig:qaoa_1p}. Increasing the system size to 16 qubits does not seem to affect the noise mitigation capability.
\begin{figure*}[ht!]
        \centering
        \includegraphics[width=\textwidth,center]{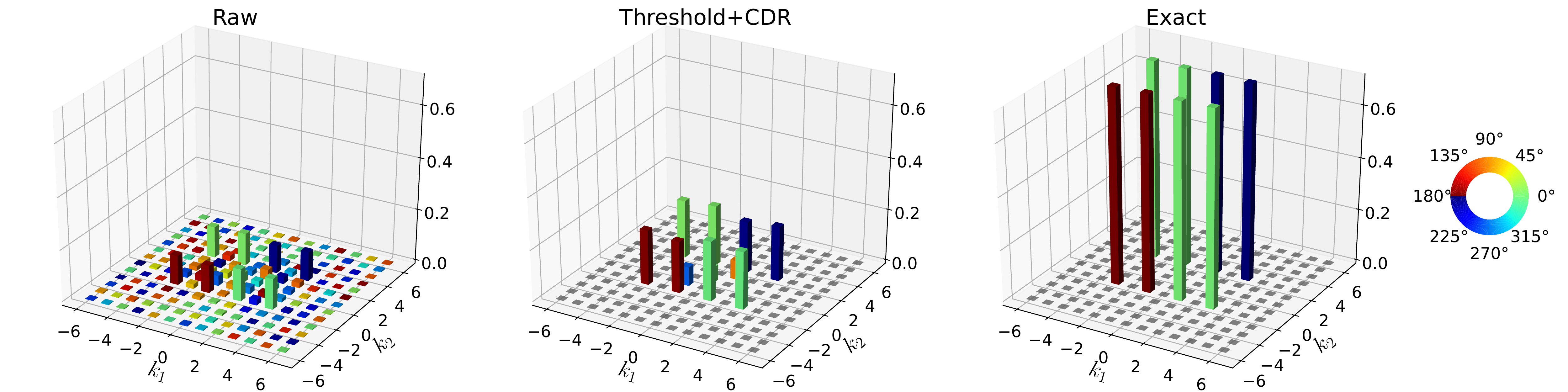}
        \caption{QAOA experiment for MaxCut of an 8-node random 3-regular graph, $p=1$, $O=H$. Magnitudes for Fourier coefficients are shown before and after filtering and CDR, vs exact. Colour represents the complex argument as shown by the colourwheel. Filtering method is thresholding with $B=10$. Run on \texttt{ibmq\_sydney}. Spectral filtering is seen to eliminate spurious frequencies, while CDR improves the magnitude of the coefficients, although only partially.}
\label{fig:qaoa_1p}
\end{figure*}

On the other hand, circuits with $p=2$ generally see little improvement in either measures of quality. This is due to the fitting subroutine in CDR failing to find an appropriate scaling factor due to Clifford data being extremely noisy. As a related effect, hardware noise causes the experimental signal to be considerably different from the noiseless one as can be seen in Figure~\ref{fig:qaoa_2p}. Overall this is a reminder that excessive noise can be an insormountable obstacle for any kind of error mitigation, whether classic or frequency-based.
\begin{figure*}[ht!]
        \includegraphics[width=\textwidth,center]{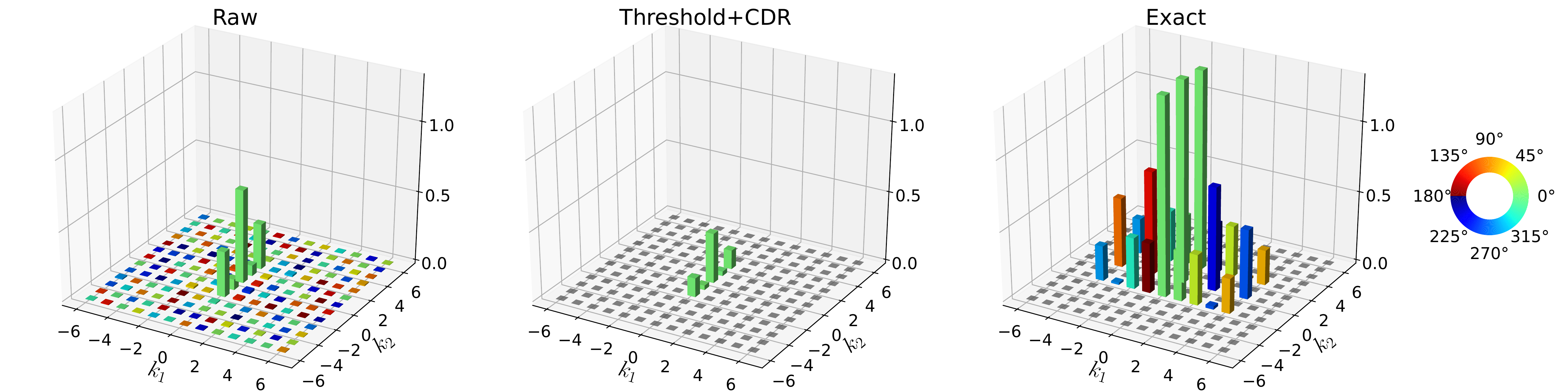}
        \caption{QAOA experiment for 8-node random 3-regular MaxCut, $p=2$, last layer params varied, first layer: $[\beta, \gamma] = [1/4, 1/4]$, $O=H$. Magnitudes for Fourier coefficients are shown before and after filtering and CDR, vs exact. Colour represents the complex argument as shown by the colourwheel. Filtering method is thresholding with $B=10$. Run on \texttt{ibmq\_sydney}. Note that both filtering and CDR fail to mitigate for noise since the raw signal is significantly different from the exact one.}
        \label{fig:qaoa_2p}
\end{figure*}

\section{Discussion}
\label{sec:discussion}

We investigated the quantum noise-induced effects on the Fourier spectrum of parameterised quantum circuits, both in real hardware and through theoretical analysis of common noise models. This led us to design methods for digital error mitigation based on filtering noise-induced modes or thresholding most-significant coefficients. For specific circuit classes, in Appendix~\ref{sec:AppendixFilter} show that
if the noise is below a certain threshold one may perfectly reconstruct the landscape from noisy circuit evaluations. While these circuits do not exhibit quantum advantage, they illustrate that denoising in the frequency domain may recover error-free estimates of expected values from hardware data without performing any classical simulation. However, these methods, among others derived from them, can be employed for more general circuits. In particular we show how to combine denoising in Fourier space with existing error mitigation based on learning from Clifford circuits. The latter targets incoherent noise, while the former reduces the effect of gate-dependent, correlated errors, including coherent contributions. We implement these methods in different quantum hardware using application-specific circuits employed in QAOA and VQE applications. The compound effect of the Fourier-based error mitigation strategies reduces not only the magnitude of errors but also the noise-induced distortion in the energy landscape. 

The most glaring disadvantage of Fourier methods is the large sampling overhead required. We discuss the resource requirements, as well as possible ways to reduce it, in Appendix \ref{app:resource}. There we discuss also how to approach the problem of incomplete or nonuniform sampling, which might also help ameliorate the resource requirement and make the methods more flexible.

While here we only explored combining spectral methods with CDR, one may adapt our approach for other techniques such as zero noise extrapolation \cite{temme2017error, ying2017efficient} and probabilistic error cancellation \cite{temme2017error, endo2018practical}, which similarly address decoherence effects. More recent work \cite{czarnik2020error} that improves the efficiency of CDR could also be employed to reduce the total number of training circuits required. Furthermore, investigating the effect of noise on the Fourier coefficients might also give new insight into the extent to which hardware noise satisfies some of the in-built assumptions of error mitigation protocols. For example, digital versions of zero noise extrapolation involve increasing the noise level by introducing additional gates~\cite{giurgica2020digital}. Fourier methods could then be useful to determine appropriate type of such gate folding to validate the different fitting functions used in the extrapolation.

Beyond this, our theoretical contributions may also be of independent interest in building data-encoding models for quantum machine learning and investigating computational complexity of classical optimisation in variational quantum algorithms. Recent work \cite{bittel2021training} shows that the classical optimisation step for variational quantum algorithms is NP-hard. Essentially this relies on the NP-hardness of finding extremal points of specific classes of trigonometric polynomials and constructing corresponding circuits. Since our work explores this connection between trigonometric polynomials and circuit structures, it may also be useful to investigate worst-case or even average-case hardness in the optimisation step for more restricted classes of ans\"atze.

A major conceptual contribution of this work is reframing the study of quantum landscapes in the familiar terms of spectral analysis. As such, the theory developed here may be a fruitful approach to the recently coined quantum landscape theory~\cite{larocca2021theory}. Spectral analysis of the landscape can be employed to make statements about its global characteristics; for example, the presence or absence of barren plateaus and narrow gorges~\cite{cerezo2021cost}. Furthermore, within this framework noise channels can be included naturally, making it amenable to study phenomena such as noise-induced barren plateaus~\cite{wang2020noise} and alterations of the quantum training landscape~\cite{wang2021can}.

Our work also opens up new directions to explore the vast range of signal processing techniques for the purpose of quantifying and reducing quantum noise. We tackled these topics from the perspective of digitised quantum computations, however similar analysis may be applied at the level of pulse sequences. In fact, the latter connects with experimental techniques used in the tuning and calibration of individual gates such as noise spectroscopy~\cite{alvarez2011measuring, yuge2011measurement, young2012qubits}.

In this work, we focused on parameterised circuits typically used in NISQ applications. The parameter domain and the gates employed by such circuits make it natural to employ $U(1)$ Fourier transforms. However, one can also consider families of circuits with structures that are embeddings of different groups, in which case one might look at Fourier transforms with respect to finite abelian (or even non-abelian) groups.
For example, randomised benchmarking involves randomly sampling gates in the Clifford group to form sequences of increasing depth, making it suitable to use discrete spectral methods to analyse the effect of gate dependent errors~\cite{merkel2018randomized}. Moreover, many algorithms such as phase estimation involve repeated applications of similar circuit structures, so a noise analysis in the Fourier domain may also prove to be fruitful.
Overall the work presented here introduces a flexible framework that might extend to numerous other methods of noise characterisation and mitigation, but also to the analysis of (variational) quantum algorithms.

\section{Acknowledgements}
\label{sec:acknowledgements}

The authors would like to thank Steven Herbert, Nathan Fitzpatrick and David Muñoz Ramo for useful feedback and discussions. We thank Dan Mills and Steven Herbert for helpful comments on the final draft of this manuscript.
EF and IR acknowledge the support of the UK government department for Business, Energy and Industrial Strategy through the UK national quantum technologies programme. EF acknowledges the support of an industrial CASE (iCASE) studentship, funded by the UK Engineering and Physical Sciences Research Council (grant EP/T517665/1), in collaboration with the university of Strathclyde, the National Physical Laboratory, and Cambridge Quantum Computing. Code and data used in this project is available at \url{https://github.com/e-font/fourier-noise}.

\bibliography{paper}

\clearpage
\onecolumngrid
\appendix
\vspace{0.5in}

\begin{center}
  { \bf \MakeUppercase{Appendix}}
\end{center}

\section{Clifford path decomposition results} 
\label{appendix:Cliff}

The unitary channel corresponding to $Z$-rotation by $\theta$ angle, $Z(\theta) = e^{- i \theta_j Z/2}$, is given by 
\begin{equation}
\label{eq:A1}
 {\cal{Z}}(\theta_j)(\rho) = p_{\cal{I}}(\theta_j) \rho + p_{\cal{Z}}(\theta_j) Z\rho Z\hc + p_{\cal{S}}(\theta_j) S\rho S\hc
\end{equation}
with $Z$ the Pauli operator, $S=\sqrt{Z}$ and quasi-probabilities given by
\begin{align}
        p_{\cal{I}}(\theta) & = \frac{1}{2} (-\sin{\theta} + \cos{\theta}+1)\\
        p_{\cal{Z}}(\theta) & = \frac{1}{2} (-\sin{\theta} - \cos{\theta}+1)\\
        p_{\cal{S}} (\theta) & = \sin{\theta}
        \label{appendix:coeffZ}
\end{align}
and similarly,
\begin{equation}
         {\cal{Z}}(\theta_j)(\rho) =  \mathcal{C}_0 + e^{i\theta} \mathcal{C}_{1} + e^{-i\theta} \mathcal{C}_{-1}
\end{equation}
where the process modes are given by linear combinations of Clifford unitary channels as follows
\begin{align}
        \mathcal{C}_0& = \frac{1}{2}\mathcal{I} + \frac{1}{2}\mathcal{Z}\\
        \mathcal{C}_{1}& = \frac{1+i}{4}\mathcal{I} - \frac{1- i}{4}\mathcal{Z}- \frac{i}{2}\mathcal{S} \\
        \mathcal{C}_{-1}& =  \frac{1-i}{4}\mathcal{I} - \frac{1+ i}{4}\mathcal{Z}+ \frac{i}{2}\mathcal{S}       
        \label{appendix:cliff}
\end{align}

A general Pauli gadget takes the form $U(\theta_j) = e^{-iP_j\theta_j/2}$ for some $n$ qubit Pauli operator $P_j$. Therefore, as a quantum channel $\mathcal{U}(\theta_j)(\cdot) := U(\theta_j) ( \, \cdot \, ) U(\theta_j)\hc$ it decomposes into
\begin{equation}
        \U(\theta_j) = \mathcal{C}_0^{(j)} + e^{i\theta_j} \mathcal{C}_1^{(j)} + e^{-i\theta_j} \mathcal{C}_{-1}^{(j)}
\end{equation}
where, similarly as in the decomposition of $\mathcal{Z}(\theta)$, here the Clifford process modes are
\begin{align}
        \mathcal{C}_{0}^{j} &= \frac{1}{2}\left(\mathcal{U}(0) + \mathcal{U}(\pi) \right) \\
        \mathcal{C}_{1}^{j} &=  \frac{1+i}{4}\mathcal{U}(0) - \frac{1- i}{4}\mathcal{U}(\pi)- \frac{i}{2}\mathcal{U}(\pi/2) \\
        \mathcal{C}_{-1}^{j}&= \frac{1-i}{4}\mathcal{U}(0) - \frac{1+ i}{4}\mathcal{U}(\pi)+ \frac{i}{2}\mathcal{U}(\pi/2)       
\end{align}
where (up to a global phase that does not affect the corresponding unitary channels), $U(0) = I$, $U(\pi) =P_j$ and $U(\pi/2) =  \frac{1}{\sqrt{2}}(I - i P_j)$. 
\begin{lemma}
        Suppose ${\bf{j}}_1 = (j^{1}_1,... j^1_n)$, ${\bf{j}}_2=(j^2_1,..., j^2_n)$ are binary strings with Pauli operators $P_{\bf{j}_1} = Z^{j_1^1}\otimes ...\otimes Z^{j^{1}_n}$ and similarly for $P_{\bf{j}_2}$. Then,
        \begin{equation}
                \C_{\pm 1}^{\bf{j}_1} \circ \C_{0}^{\bf{j_2}} \circ \C_{0}^{\bf{j_2}\oplus \bf{j_1}} = 0
        \end{equation}
\end{lemma}
\begin{proof}
        First, $\C_0^{\bf{j}_2} \circ \C_{0}^{\bf{j}_2\oplus \bf{j}_1} = \frac{1}{4} (\I+ \P_{\bf{j}_2} )\circ (\I+ \P_{\bf{j}_1\oplus \bf{j}_2}) =\frac{1}{4} (\I + \P_{\bf{j}_1} +\P_{\bf{j}_2} + \P_{\bf{j}_1\oplus \bf{j}_2})$. Then, we observe that $\U(\pi/2) + \U\hc(\pi/2) = \U_j(0) + \U_j(\pi)$ and furthermore $\U_j(\pi/2) \U_j(\pi) = \U_j\hc(\pi/2) $. The latter trivially follows from $3\pi/2 = 2\pi - \pi/2$ and the former follows from
        \begin{align}
                \U(\pi/2)(\rho) &= \frac{1}{2} (I-iP_j)(\rho)(I+iP_j)  \\
                &=\frac{1}{2} ( \rho  + P_j(\rho)P_j + i[\rho P_j - P_j\rho]).
        \end{align}
        Therefore, $C_{\pm 1}^{\bf{j}_1} \circ  \C_{0}^{\bf{j_2}} \circ \C_{0}^{\bf{j_2}\oplus \bf{j_1}} = \frac{1}{16} (\I - \P_{\bf{j}_{1}})\circ (\I + \P_{\bf{j}_1} +\ P_{\bf{j}_2} + \P_{\bf{j}_1\oplus \bf{j}_2}) \pm \frac{i}{4} (\frac{1}{2} (\I + \P_{\bf{j}_1} + \P_{\bf{j}_2} +\P_{\bf{j}_1\oplus\bf{j}_2} ) - \frac{1}{2}\U_{j_1}(\pi/2)) \circ (\I + \P_{\bf{j}_1} +\P_{\bf{j}_2} + \P_{\bf{j}_1\oplus \bf{j}_2})$. The first term is zero due to the permutations, whilst the third becomes $\U_{j_1}(\pi/2) + \U_{j_1}\hc(\pi/2)  + \U_{j_1}(\pi/2) \circ\P_{\bf{j}_2} +\U_{j_1}\hc(\pi/2)\circ \P_{\bf{j}_2} = \I+ \P_{\bf{j}_1}  +  (I +\P_{\bf{j}_1}) \circ \P_{\bf{j}_2} =  \I + \P_{\bf{j}_1} + \P_{\bf{j}_2} + \P_{\bf{j}_1\oplus \bf{j}_2}$ and therefore cancels out the second term giving the result.  
\end{proof}

\section{Trigonometric form of expectation values}
\label{app:trig}
First consider the following Lemma:

\begin{lemma}[Transformation of Pauli operators by Z-rotations]
\label{lem:ZRotTransformation}
Given a conjugate Z-rotation channel $\mathcal{Z}^\dagger_\theta$ acting on qubit k, the result of its application onto a Pauli operator $P$ is
\begin{equation}
\mathcal{Z}^\dagger_\theta(P) = \begin{cases}
                                                                                P & \text{if}\; \sigma^{(k)} = \sigma_0, \sigma_3, \\
                                                                                \cos(\theta)P \pm \sin(\theta)P' & \text{otherwise}.\\
                                                                                \end{cases},
\end{equation}
where $P'$ is another Pauli operator and $\sigma^{(k)}$ denotes the $k$-th single-qubit Pauli composing the operator $P$.
\end{lemma}

\begin{proof}
Let us focus on a single-qubit Pauli operator $\sigma_i$. Equation~\ref{eq:A1} described the result of applying the conjugate Z channel onto this operator.
It is easy to show that $\mathcal{Z}(\sigma_i) = \mathcal{S}^\dagger(\sigma_i) = \sigma_i$ if $i = 0, 3$, and that $\mathcal{Z}(\sigma_i) = -\sigma_i$ for $i = 1, 2$, $\mathcal{S}^\dagger(\sigma_1) = -\sigma_2$, $\mathcal{S}^\dagger(\sigma_2) = \sigma_1$. Therefore, we get the following cases:
\begin{equation}
\mathcal{Z}^\dagger_\theta(\sigma_i) = \begin{cases}
                                                                                \sigma_i & \text{if}\; i = 0, 3, \\
                                                                                \cos(\theta)\sigma_i \pm \sin(\theta)\sigma_i' & \text{otherwise}.\\
                                                                                \end{cases},
\end{equation}
where we labelled $\sigma_i'$ as the unsigned Pauli operator one gets after applying the $\mathcal{S}^\dagger$ channel.

As the Z-rotation channel only acts on a single qubit, we can in general write the action on any Pauli operator in the same form, with the outcome depending on the single-qubit Pauli that the rotation acts upon. Therefore, the result follows.
\end{proof}

This can be used to prove the following result:

\begin{lemma}[Trigonometric form of expectation values]
        Consider a parameterised quantum unitary of the form 
        \begin{equation}
                U(\boldsymbol\theta) = V_0 U_1(\theta_1) V_1 U_2(\theta_2) V_2 \, ...\, U_m(\theta_m)V_m
        \end{equation}
        and the state resulting from applying the corresponding channel to a fixed initial state: $\rho(\boldsymbol\theta) =  \U(\boldsymbol\theta)(\rho_0)$.
        Now suppose each of the one-parameter unitaries is generated by Pauli operators $U(\theta_j) = e^{-i P_j \theta_j/2}$, and in addition the unparameterised unitaries $V_i$ are Clifford and the input state $\rho_0$ is a stabiliser state. 
        Then the Fourier coefficients of the trigonometric polynomial $\Tr(\rho(\boldsymbol{\theta})P)$ where $P$ is a Pauli operator are 
        \begin{equation}
                c_{\bf{k} }  \in \{\pm 1/2^{||{\bf{k}}||_1}, \pm i 1/2^{||{\bf{k}}||_1}, 0\} .
        \end{equation} 
        where ${\bf{k}} \in \{-1,0,1\}^{m}$.
        Furthermore, the expansion of $\tr(\rho(\boldsymbol{\theta}) P)$ in terms of trigonometric functions becomes particularly simple:
        \begin{equation}
                \tr(\rho(\boldsymbol{\theta}) P) = \sum_{\bf{k}}  c'_{\bf{k}} \phi_{\bf{k}} (\boldsymbol{\theta})
        \end{equation}
        with $c'_{\bf{k}} \in \{0, \pm 1\}$ and $\phi_{\bf{k}}(\boldsymbol{\theta}) = \phi_{k_1}(\theta_1) \phi_{k_2}(\theta_2) ... \phi_{k_m}(\theta_m)$ where $\phi_{0}(\theta)=1$, $\phi_{1}(\theta) = \cos{\theta}$ and $\phi_{-1} (\theta) =\sin{\theta}$.
\end{lemma}

\begin{proof}[Proof]
        It is most direct to prove first that the expectation value $\tr(\rho(\boldsymbol{\theta}) P)$ is in the given trigonometric form. To do this, first consider that every parameterised operator $U(\theta_j) = e^{-i P_j \theta_j/2}$ is a Pauli gadget, and as such it can be decomposed into one single-qubit parameterised Z-rotation and Clifford gates, which can be absorbed into the fixed gates. Therefore, the circuit consists entirely of fixed Clifford gates $V'$ and parameterised single-qubit Z-rotations:
        \begin{equation}
                U(\boldsymbol\theta) = V'_0 Z_1(\theta_1) V'_1 Z_2(\theta_2) V'_2 \, ...\, Z_m(\theta_m)V'_m.
        \end{equation}
        Now consider the channel form of the unitary as it appears in the calculation of the expectation value $\tr(P \U(\boldsymbol\theta)(\rho_0))$, which in the Heisenberg picture becomes $\tr(\U^{\dagger} (\boldsymbol\theta) (P) \rho_0)$.
        Every step of the calculation of the expectation value then consists in the successive application of either a conjugate Z-rotation channel or a Clifford channel. Clifford channels by definition normalise the Pauli group and hence turn a Pauli operator into another Pauli operator. On the other hand according to Lemma~\ref{lem:ZRotTransformation} Z-rotation channels may, depending on the operator, produce a linear combination of operators with trigonometric coefficients.
        Therefore, at the end of the procedure, one is left with a sum of Pauli operators
        \begin{equation}
                \U^{\dagger} (\boldsymbol\theta) (P) = \sum_{\bf{k}\in \boldsymbol{\Lambda}} \pm \phi_{\bf{k}} P_{\bf{k}}
        \end{equation}
        where $\boldsymbol{\Lambda}=\{-1, 0, 1\}^m$ and
        \begin{equation}
                \phi_{\bf{k}}(\boldsymbol{\theta}) = \phi_{k_1}(\theta_1) \phi_{k_2}(\theta_2) ... \phi_{k_m}(\theta_m)
        \end{equation}
        where $\phi_{0}(\theta)=1$, $\phi_{1}(\theta) = \cos{\theta}$ and $\phi_{-1} (\theta) =\sin{\theta}$.
        As a final step, the expectation value of this operator is taken with respect to the stabiliser $\rho_0$ state, which for each term yields either $0$ or $\pm 1$. One is therefore left with a function of the wanted form.

        To express the trigonometric function as a Fourier series, each trigonometric monomial can be expanded in complex exponentials, giving
        \begin{equation}
                \tr(P \U(\boldsymbol\theta)(\rho_0) = \sum_{\boldsymbol{k} \in \boldsymbol{\Lambda}} c_{\boldsymbol{k}} e^{i\boldsymbol{\theta}\cdot\boldsymbol{k}}.
        \end{equation}
        where the coefficients are $c_{\bf{k}} \in \{\pm 1/2^{||{\bf{k}}||_1}, \pm i 1/2^{||{\bf{k}}||_1}, 0\}$, where $||\boldsymbol{k}||_1 = \sum_i |\boldsymbol{k}_i|$. As the function is real, $c_{\boldsymbol{-k}} = c^*_{\boldsymbol{k}}$.
\end{proof}

\section{Proof of Theorem~\ref{thm:1}}
\label{sec:proof-thm-1}

From Ref.~\cite{welch2014efficient} we have that any diagonal operator has a decomposition in terms of the diagonal Pauli operators acting only with $Z$ or $I$
\begin{equation}
        \sum_z \lambda_z |z\>\<z| = \sum_{{\bf{j} }= (j_1,..., j_n)} a_{\bf{j}} Z^{j_1} \otimes ...\otimes Z^{j_n}.
\end{equation}
Each $j_{i}$ is a binary variable with ${\bf{j}} := (j_1, ..., j_n)$ and the vector ${\bf{a}} := (a_0, a_1,..., a_{2^{n}-1})$ is the Walsh-Hadamard transform of the eigenvalue vector ${\boldsymbol{\lambda}} = \left(\lambda_z :  z\in\{0,1\}^n\right)$ so we write component-wise
\begin{equation}
        a_{\mathbf j} = \frac{1}{2^n} \sum_{z=0}^{2^n-1} \lambda_z (-1)^{\mathbf{z}\cdot\mathbf{j}}
\end{equation}
where ${\bf{z}}\cdot{\bf{j}} = \sum_{i=1}^{n} z_i \, j_i  $ and with inverse
\begin{equation}
        \lambda_{z} = \sum_{\bf{j}} a_{\bf{j}} (-1)^{\bf{z}\cdot \bf{j}}
\end{equation}
Therefore, the action of the unitary
\begin{equation}
        U(\theta) = e^{iH \theta} = W \prod_{\bf{j}} e^{i a_{\bf{j}} Z^{j_1} \otimes... \otimes Z^{j_n} \theta} W\hc
\end{equation} 
can be decomposed into a series that involves replacing the action of each phase gadget $Z^{\bf{j}}(\theta) = e^{ia_{\bf{j}} Z^{j_1}\otimes ... \otimes Z^{j_n}\theta}$ by a Clifford process modes decomposition
\begin{equation}
        \mathcal{Z}^{\bf{j}} (\theta) =  \mathcal{C}_{0}^{\bf{j}} + e^{2i\theta a_{\bf{j}}} \mathcal{C}_{-1}^{\bf{j}} + e^{-2i\theta a_{\bf{j}}} \mathcal{C}_{1}^{\bf{j}}.
\end{equation}
This is similar to the decomposition in the example from Sec~\ref{sec:Example} and the exact form of each superoperators $\mathcal{C}_{k}^{\bf{j}}$ for $k\in \{-1,0,1\}$ may be found in Appendix~\ref{appendix:Cliff}.
Therefore,
\begin{equation}
        \mathcal{U}(\theta) = \sum_{{\bf{k}} = (k_1,..., k_{r}) } e^{-2i\theta \bf{a}\cdot \bf{k}} \, \mathcal{W}\circ \mathcal{C}_{k_1}^{{\bf{j}}_1} \circ ... \circ \mathcal{C}_{k_r}^{{\bf{j}}_r}\circ \mathcal{W}\hc
        \label{eqn:cliffdec}
\end{equation}
where ${\bf{a}\cdot \bf{k}} = \sum_{i=1}^{r} a_{{\bf{j}}_i} k_{i}$   with $r$ the rank of the vector $(a_0,..., a_{{2^n}-1})$ and $\{{\bf{j}}_{1},..., {\bf{j}}_{r} \} = \{{\bf{j}}: a_{\bf{j}} \neq 0 \}$. Each vector $\bf{k}$ has entries in $\{-1,0,1\}$ and labels a particular Clifford path decomposition.

The decomposition of $\mathcal{U}$ in the form of Equation~\ref{eqn:cliffdec}, under certain circumstances, may in fact naively overcount the number of frequencies $\omega_{\bf{k}} := -2\bf{a}\cdot{\bf{k}}$. Interestingly, this comes from the fact that phase gadgets commute and particular compositions of process modes of the type $\mathcal{C}_{k}^{\bf{j}}$ give a null channel, so the corresponding term in Equation~\ref{eqn:cliffdec} vanishes. For example, $\mathcal{C}_{-1}^{\bf{j}_1}\circ \mathcal{C}^{{\bf{j}}_2}_{0}\circ \mathcal{C}^{{\bf{j}}_1\oplus {\bf{j}}_2}_{0}= \bf{0}$ with more details in the Appendix~\ref{appendix:Cliff}. This means there are certain `selection'-type of rules for the vectors $\bf{k}$, so the summation in Equation~\ref{eqn:cliffdec} ranges only over a subset of the full lattice of points $\mathbb{Z}_{3}^{\times r}$. In particular, they will take the form $k_{i} = \frac{(-1)^{{\bf{j}}_{i}\cdot \bf{z}} - (-1)^{{\bf{j}}_{i}\cdot \bf{z'}}}{2} $ for $i\in \{1,2,..., r\}$ for some binary vectors $\bf{z}, \bf{z'}$ of length $n$ and with corresponding frequencies $\omega_{\bf{k}} = \lambda_{\bf{z}} -\lambda_{\bf{z'}}$.

In the above analysis, we considered the sequence of phase gadgets as it is, and we were not a priori synthesising the unitary into a circuit containing a minimal number of single qubit $Z$-rotations. This is also one of the reasons we had to employ `selection' rules for the $\bf{k}$ vectors as in some sense we had redundancies in our description. Instead, if we were to start from a given optimised circuit and applied the Clifford decomposition directly to the parameterised gates in said circuit, then we would not come across the situation where a large number of channel terms are in fact zero. This highlights two aspects. Firstly, that there may be a connection between circuit synthesis and the fact that particular compositions of the process modes considered give a zero channel. Secondly, while employing a decomposition as in Equation~\ref{eqn:cliffdec} recovers the previously known result that frequencies correspond to differences between any two eigenvalues of the generator, our present approach also offers a route to determine those frequencies when the unitary is described in circuit form.

\section{Proof of Theorem~\ref{thm:noisy-expectations}}
\label{sec:proof-thm-2}

\begin{theorem}
Writing $\tilde{\rho}(\theta) = \widetilde{\U}(\boldsymbol{\theta})(\tilde{\rho}_0)$ for the noisy implementation of the target state $\rho(\boldsymbol{\theta})$ we have
\begin{equation}
        \<O\>_{\tilde{\rho}(\boldsymbol{\theta})} = \sum_{\bf{k}\in \tilde{\boldsymbol \Lambda}}  \tilde{c}_{\bf{k}} e^{i \tilde{\boldsymbol{\omega}}_{\bf{k}} \cdot \boldsymbol{\theta}}
\end{equation}
where now the set of frequencies $\tilde{\omega}_{\bf{k}}$ are no longer necessarily indexed with a discrete subset of an integer lattice bounded by 1.
\end{theorem}
\begin{proof}
One may formally write $\tilde{\U}(\boldsymbol{\theta}) = \mathcal{E} \circ \U(\boldsymbol{\theta})$, where the noise channel $\E$ may itself be dependent on the angles $\boldsymbol{\theta}$. One may then employ decompositions as in Lemma~\ref{lemma:decompchannel} for $\mathcal{E}$ so that only the quasiprobability distribution carries an angle dependency, namely
\begin{equation}
        \E = \sum_{j} q_j(\boldsymbol{\theta}) \, \mathcal{C}_j + \sum_{j}  q_j' (\boldsymbol{\theta}) \S_j.
\end{equation}
Therefore the noisy expected values take the form
\begin{equation}
        \<O\>_{\tilde{\rho} (\boldsymbol{\theta})} = \sum_{j} q_{j}(\boldsymbol{\theta}) \Tr[\rho({\boldsymbol{\theta}}) \mathcal{C}_{j}\hc(O)] + \sum_j q_{j}'(\boldsymbol{\theta}) \<P_j|O|P_j\>
\end{equation}
where we have used that the adjoint Pauli reset channel is $\S_{j}\hc (O) = \mathbb{I} \, \<P_j|O|P_j\>$  for some eigenstate $|P_j\>$ of a Pauli operator. Following the analysis in the previous section, 
$\Tr[\rho({\boldsymbol{\theta}}) \mathcal{C}_{i}\hc(O)]  = \sum_{\bf{k} \in \boldsymbol \Lambda} c_{\bf{k}} (\mathcal{C}\hc_i(O)) e^{i\boldsymbol{\omega}_{\bf{k}}\cdot {\boldsymbol{\theta}}}$ where we explicitly add the dependency of the coefficients $c_{\bf{k}}$ on the observable. This implies that
\begin{equation}
                \<O\>_{\tilde{\rho} (\boldsymbol{\theta})} = \sum_{\bf{k}\in \bf{\Lambda}} \sum_{i} q_{i}(\boldsymbol{\theta}) c_{\bf{k}}(\mathcal{C}_i\hc(O)) e^{i\boldsymbol{\omega}_{\bf{k}}\cdot \boldsymbol{\theta}} + \sum_j q_{j}'(\boldsymbol{\theta}) \<P_j|O|P_j\>.
                \label{eqn:generalnoise}
\end{equation}
The second summation can be absorbed into the $\omega_{\boldsymbol{k}} = 0$ term giving the intended result.

Let us denote by $\mathcal{S} $ the (maximal) frequency support for the Fourier decomposition of the unitary channel $\U(\boldsymbol{\theta})$. For any given observable, the frequency spectrum of $\<O\>_{\rho(\boldsymbol{\theta})}$ will be included in $\mathcal{S}$, but may contain fewer modes. As we have seen in the previous section, for each independent parameter $\theta_i$ the  frequency spectrum is included in a finite lattice  $\mathcal{S}_i:=\{\omega_{{\bf{k}}_i} = -2{\bf{a}}_i\cdot {\bf{k}}_i  :  {\bf{k}}_{i} \in {\mathbb{Z}}_3^{\times r_i}   \}$ over $\mathbb{Z}_3$ with real basis vectors $-2{\bf{a}}_i$ which can be determined from the circuit description of $U_i(\theta_i)$.
On the other hand, there's no a priori constraint for $q_i, q_i'$ to have discrete spectra since the noise channel may not present periodicity in the angles. 
\end{proof}

\section{Resource requirements and nonuniform sampling}
\label{app:resource}

As mentioned in the text, sampling overhead is the largest obstacle to a pratical implementation of spectral analysis methods. One-dimensional DFT up to a frequency of $N$ would require sampling $d = 2N+1$ points, itself a linear scaling. However, the scaling in the dimension is unfavourable as it is exponential, practically limiting the applicability of the method to small sets of parameters. In a practical setting the runtime is also affected by the requirement of high shot counts to reduce shot noise. In itself, however, the method could still be useful even in these conditions. A variational circuit with several parameters could be broken down into sets of parameters which then could be mitigated for noise individually. Furthermore, one might expect that different parameterised gates would be affected differently by noise based on factors such as their order and the noise characteristics of the qubits they act on. For example the effect of the first gates in the circuit would be diluted by successive noise channels. Therefore, one might choose to focus on the most noisy gates in order to maximise the benefit of mitigation.

In addition, we remark that many variational algorithms of interest do not feature many individual parameters, but rather a few parameters that govern larger sets of correlated gates. In that case, the sampling points would be much reduced, even when accounting for the larger set of frequencies that need to be sampled. For example, while the spectrum of $M$ uncorrelated single-qubit rotations would have size $3^M$, correlating the angles into a single parameter would require sampling a number of points linear in $M$.

Finally, the use of advanced signal extraction techniques might enable these methods to work with much reduced sets of data. There exists an array of techniques that allow for the retrieval of the most significant Fourier coefficients from a signal, using a subexponential amount of samples~\cite{candes2006stable, tropp2007signal, donoho2006compressed, candes2006robust}. In the case of sparse signals, techniques such as the sparse fast Fourier transform (sFFT) actually allow for complete reconstruction with sampling much below the Nyquist rate~\cite{hassanieh2012simple}. Therefore, if the quantum signal consists of a few large Fourier coefficients or, even better, if it is indeed sparse, these techniques could allow to circumvent the exponential sampling requirements in the case of many parameters. Interestingly, they would also apply an implicit thresholding by selecting only the most significant coefficients, and therefore is expected to provide some denoising effects.

\subsection{Non-integer and non-equidistant samples or frequency
  spectra}
\label{sec:non-integer-non}

In some cases, it might be desirable to get an estimate of the signal at non-integer frequencies, or spectra that are not equidistant. This might be relevant in case, for example, the error channel is suspected to depend in a nontrivial way on a parameter, leading to frequencies that are non-integer multiples of the parameter's frequency. This problem is known as non-uniform DFT or NUDFT and various approaches have been developed that implement an efficient version (NUFFT), based on interpolation~\cite{boyd1992fast, dutt1993fast, potts2003fast} or low-rank approximations~\cite{ruiz2018nonuniform}.\\
The discrete Fourier transform can also be implemented on incomplete data or nonuniform sampled data, while returning frequency data on a uniform support, known as type II NUDFT~\cite{ruiz2018nonuniform}. Fast algorithms exist also in this scenario, which similarly use approximations and therefore come at the cost of reduced accuracy~\cite{ruiz2018nonuniform, potts2003fast}, however might be sufficient as a rough error estimation technique. There exist scenarios where a large amount of nonuniform data points are generated as a byproduct, for example in the optimisation of a variational algorithm. Using nonuniform estimation, it might be possible to implement spectral noise evaluation at no additional sampling cost.\\
Finally, we note the connection between the spectral analysis of quantum algorithms and signal processing in NMR. The latter similarly features a costly sampling procedure that generates multidimensional data, which then needs to be processed into its spectral components. Techniques such as nonuniform sampling~\cite{maciejewski2011data} and maximum entropy reconstruction~\cite{hoch2014nonuniform} have been devised to minimise the sampling and therefore the running time, while reconstructing the image with accuracy. One can therefore envisage directly utilising similar techniques for quantum algorithm.

\section{Effect of different channels on Fourier modes}
\label{app:noisesFourier}

\begin{lemma}[Depolarising channel]\label{lem:dep}
The noisy expectation value is
\begin{equation}
        \<O\>_{\tilde{\rho}(\boldsymbol{\theta})} = \sum_{\bf{k} \in \boldsymbol \Lambda} (1-p)^m c_{\bf{k}} e^{i\omega_{\bf{k}}\cdot \boldsymbol{\theta}} +  \frac{(1-(1-p)^m)}{2^n}\Tr(O).
\end{equation}
Therefore, the noisy Fourier coefficients $\tilde{c}_{\bf{k}}$ are homogeneously contracted.
\end{lemma}
\begin{proof}
  $\tilde{\U} = \mathcal{E}\circ \U$ where $\mathcal{E}(\rho) =  (1-p)^{m+1} \rho + (1- (1-p)^{m+1}) \frac{\iden}{2^n}$. Then, using Equation~\ref{eqn:generalnoise} the result follows.
\end{proof}

Consider a single Pauli noise channel $\mathcal{N}_{p} := (1 - p)\mathcal{I} + p \mathcal{P}$, which appears within a quantum process
\begin{equation}
        \tilde \U = \V_{0} \circ \U_1(\theta_1) \circ \cdots \circ \mathcal{N}_{p} \circ \cdots \circ \U_{m}(\theta_m) \circ \V_m.
\end{equation}
For this we have the following
\begin{lemma}[Pauli channel]\label{lem:pauli}
Under a Pauli channel, the frequency spectrum does not change. For the specific case of $O$ a Pauli operator, the signal's coefficients are contracted: $|\tilde{c}_{\bf{k}}(O)| \le |c_{\bf{k}}(O)|$, and so there are no extra coefficients. For a general $O$, some new modes that were previously zero due to cancellation of paths might find themselves with a small non-zero value under noise.
\end{lemma}
\begin{proof}
First we consider the case of $p=1$, that is, a Pauli error $P$ that occurs within the quantum process
\begin{equation}
        \tilde \U = \V_{0} \circ \U_1(\theta_1) \circ \cdots \circ \mathcal{P} \circ \cdots \circ \U_{m}(\theta_m) \circ \V_m.
\end{equation}
Assume that each $\U_{i}(\theta_i) = e^{-iP_i\theta_i/2}$ where $P_i$ are Pauli operators themselves, and that the unparameterised unitaries are Clifford as well. Then since Pauli operators either commute or anti-commute then propagating a Pauli error $P$ through the circuit will change a subset of $\theta_i \to -\theta_i$.  In terms of the individual Clifford paths $\mathcal{V}_0 \circ \mathcal{C}^{1}_{k_1} \circ ... \circ \mathcal{C}^{m}_{k_m}\circ \V_m$ labelled by vector ${\bf{k} }= (k_1,..., k_m) \in \mathbb{Z}_3^{\times m}$, the effect of a Pauli error within the circuit will result in a "noisy" path where $\mathcal{P}'\circ \mathcal{V}_0 \circ \mathcal{C}^{1}_{k'_1} \circ ... \circ \mathcal{C}^{m}_{k'_m}\circ \V_m$ labelled by vector ${\bf{k}' }$ for which $k'_i = - k_i$ on the fixed (path-independent) subset of $i \in \{1,...,m\}$. What happens is that since the lattice $\mathcal{S}$ contains both $\omega$ and $-\omega$ then this means that the spectrum of frequencies for noisy $\tilde{\U}$ will, just like those for $\U$, also be included in $\mathcal{S}$. However, some new modes  that were previously zero due to cancellation of paths might find themselves with a small non-zero value under noise for specific target observables. For the corresponding Fourier coefficient $c_{\bf{k}}(O)$ for the expected value of operator $O$ over some input state, its noisy version will be $\tilde{c}_{\bf{k}}(O) = c_{\bf{k'}}(\P'(O))$.
Similarly, in the case of a Pauli channel $\N_{p} = (1-p)\I + p\P$ the noisy Fourier coefficients become $\tilde{c}_{\bf{k}}(O) = (1-p) c_{\bf{k}}(O) + p c_{\bf{k'}}(\P'(O))$.

Going back to the case where $O$ is a Pauli operator, we have $\P'(O) = \pm O$. Now, $c_{\bf{k'}}$ is equal to $c_{\bf{k}}$ or its conjugate. To see this, consider any one of the components, say $k_j$. For the case of $k_j = \pm 1$, one has
\begin{align}
        c_{\bf{k}}(O)   &= \Tr[O \mathcal{V}_0 \circ \cdots \circ \mathcal{C}^{j}_{\pm1} \circ \cdots \circ \V_m (\rho)]\\
                                &= \Tr[O \mathcal{V}_0 \circ \cdots \circ \mathcal{C}^{j}_r \circ \cdots \circ \V_m (\rho)] \\
                                & \  \ \pm i \,\Tr[\O \mathcal{V}_0 \circ \cdots \circ \mathcal{C}^{j}_i \circ \cdots \circ \V_m (\rho)] \nonumber
\end{align}
where $\C^{j}_r = \frac{1}{4}\left(\mathcal{U}(0) - \mathcal{U}(\pi) \right)$ and $\C^{j}_i = \frac{1}{8}\left(\mathcal{U}(0) - 2\mathcal{U}(\pi/2) + \mathcal{U}(\pi) \right)$ are two real channels.
Therefore, $k_i \to -k_i$ is equivalent to conjugating the coefficient, and by induction, it is so for any $\bf{k'}$ as long as the number of affected nontrivial indices is odd.
It follows that $|c_{\bf{k'}}(\P'(O))| = |c_{\bf{k}}(O)|$. Therefore by the triangle inequality the noisy coefficients will obey $|\tilde{c}_{\bf{k}}(O)| \le |c_{\bf{k}}(O) |$.
\end{proof}

Thus we see that a contraction of the Fourier coefficients is a feature of Pauli channels, and this result extends the previous one on global depolarising channels.

The case of general $\V$s, as well as noise channels that are linear combinations of different Pauli channels, and multiple Pauli channels, follow analogously from these considerations. 
Furthermore, state preparation and measurement (SPAM) errors can also be modelled by Pauli noise channels.

A special case of the above for which more can be said is the aligned Pauli error, which is when a channel is followed by a Pauli channel which is identical to its generator.
\begin{lemma}[Aligned Pauli error]\label{lem:aligned}
For any operator $O$, an aligned Pauli error leads to a contraction of some coefficients without extra modes. Contrary to the previous result, this holds in the presence of non-Clifford non-parameterised gates.
\end{lemma}
\begin{proof}
Take the operator $U(\theta_j) = e^{-iP_j\theta_j/2} = \mathcal{C}^{(j)}_0 + e^{i\theta_j}\mathcal{C}^{(j)}_1 + e^{-i\theta_i}\mathcal{C}^{(j)}_{-1}$. If this channel is followed by a Pauli channel which is identical to its generator, the overall channel is $P_jU(\theta_j) = U(\theta_j + \pi) = \mathcal{C}^{(j)}_0 - e^{i\theta_j}\mathcal{C}^{(j)}_1 - e^{-i\theta_i}\mathcal{C}^{(j)}_{-1}$, and therefore for a Pauli noise channel we get
\begin{equation}
        \mathcal{P}_j\U_{i}(\theta_i) = \mathcal{C}^{(i)}_0 + (1-2p)e^{i\theta_i}\mathcal{C}^{(i)}_1 + (1-2p)e^{-i\theta_i}\mathcal{C}^{(i)}_{-1}. 
\end{equation}
Seen from the spectrum's perspective, there will be a contraction of some coefficients without extra modes.
\end{proof}

\begin{lemma}[Non-unital channel]\label{lem:nonunital}
For non-unital channels, the Fourier coefficients will be modified and new modes might be present that only include a subset of the parameters following the channel.
\end{lemma}

\begin{proof}
As explained in Section~\ref{sec:IIIA}, Clifford channels are not sufficient to represent arbitrary non-unital channels such as amplitude damping. One must also include Pauli reset channels, which trace out a qubit and initialise it with a specific eigenstate of a Pauli operator.
Therefore, from the Clifford path point of view such channels will include additional paths that are independent of the initial state, and only depend on gates following the channel
\begin{align}
\Tr[O\tilde{\U}(\rho_0)]        =       &\Tr[O\V_{0} \circ \cdots \circ \mathcal N_{U} \circ \cdots \circ \V_m (\rho_0)] \nonumber\\
                                                        &+\Tr[O\V_{0} \circ \cdots \circ \V_i (\sigma)]
\end{align}
where $\mathcal N_{NU}$ is the non-unital channel, while $\mathcal N_{U}$ is the unital part and $\sigma$ is the state which the channel resets to. The latter two need not be physical, however they still allow a description in terms of Clifford paths. 
\end{proof}

Finally, we explore what happens if the noise channels now depend on some of the parameters. 
\begin{lemma}[Parameter-dependent and correlated channels]\label{lem:corr}
For channels which depend on some parameters, we can have $\tilde\S \not\subset \S$.
\end{lemma}
This is in contrast with all other types of noise explored previously.

\begin{proof}
For simplicity, consider a single noise channel that depends on a single parameter elsewhere in the circuit.
Due to causality, the noise channel must follow the unitary with that parameter
\begin{equation}
        \tilde{\U} = \V_{0} \circ \cdots \circ \tilde \N(\theta_k) \circ \U_{l}(\theta_l) \circ \cdots \circ \U_{k}(\theta_k) \circ \cdots \circ \V_m
\end{equation}
For noise channels and unitaries acting on a subset of the qubits, it is not guaranteed that the measurement operator falls into both light cones. Let us analyse the two distinct cases. In case it does, there may be Clifford paths which depend on both channels, and therefore the range of frequencies will be increased to include the additional ones from the noise channel. Denoting by $\omega$ the frequencies in the clean spectrum $\S$, and by $\tilde\omega_{k}$ the frequencies from $\tilde\N(\theta_k)$, the new spectrum will be $\tilde\S=\{\omega + \tilde\omega_{k}\}$.
In case it does not, the two channels will not interact, and therefore the frequencies are not added together, and the new spectrum will be $\tilde\S = \S \cup \{\tilde\omega_{k}\}$.
In either case we have $\tilde\S \not\subset \S$.
\end{proof}

\section{Shot noise}
\label{sec:shot-noise}

\begin{lemma}[Shot noise]
Given a number of shots $n_s$, each Fourier coefficient is normally distributed around the noiseless mean with standard deviation
\begin{equation}
        \sigma(\Re\,\mathcal{F}[\tilde{x}]_{\boldsymbol{k}})= \frac{\sqrt{d^m - ||x||_2^2}}{d^m\sqrt{2 n_s}}
\end{equation}
and similarly for the imaginary part. Note that the distributions at different $\boldsymbol{k}$'s are not independent. 
\end{lemma}
\begin{proof}
First let us review the DFT convention we use in this paper. Given an $m$-dimensional signal $\{x_{\boldsymbol{i}}\}$, its DFT is defined as
\begin{equation}
        \mathcal{F}[x]_{\boldsymbol{k}} = \frac{1}{d^m} \sum_{\boldsymbol{j} = \boldsymbol{0}} x_{\boldsymbol{j}} e^{- i \boldsymbol{\theta}_{\boldsymbol{j}} \cdot \boldsymbol{\omega}_{\boldsymbol{k}}}
\end{equation}
with this convention Parseval's theorem is $||x||^2 = d^m ||\mathcal{F}[x]||^2$.
For simplicity we assume that both the samples and the frequencies are considered on regular lattices, and we set $\boldsymbol{\theta}_{\boldsymbol{i}} = 2\pi\boldsymbol{i}/d$ and $\boldsymbol{\omega}_{\boldsymbol{k}} = \boldsymbol{k}$.
In the quantum case, the signal is corrupted by sampling noise, which follows a binomial distribution. For a sufficient number of samples $n_s$, one can show that the signal is approximately Gaussian distributed around the expectation value $x_{\boldsymbol{i}}$, with a variance $\sigma_s(x_{\boldsymbol{i}}) = \frac{\sqrt{1 - x_{\boldsymbol{i}}^2}}{\sqrt{n_s}}$. Therefore, we can write the signal as a noiseless signal plus a Gaussian noise
\begin{equation}
        \tilde{x}_{\boldsymbol{i}} = x_{\boldsymbol{i}} + G_{\sigma_s(x_{\boldsymbol{i}})} . 
\end{equation}
Since the DFT is a linear operation,
\begin{equation}
        \mathcal{F}[\tilde x]_{\boldsymbol{k}}= \mathcal{F}[x]_{\boldsymbol{k}} + \mathcal{F}[G_{\sigma_s(x)}]_{\boldsymbol k} . 
\end{equation}
Let us focus on the noise signal. This is equivalent to a Gaussian signal with standard deviation $1/\sqrt{n_s}$ scaled by a function of $x$: $G_{\sigma_s(x)} = G_{1/\sqrt{n_s}} \sqrt{1-x^2}$. The Fourier transform of the noise can therefore be written as a convolution
\begin{equation}
        \mathcal{F}[G_{\sigma_s(x)}]_{\boldsymbol k} = \left[\mathcal{F}[G_{1/\sqrt{n_s}}] * \mathcal{F}[\sqrt{1-x^2}]\right]_{\boldsymbol k}
\end{equation}

Now the following lemma can be proven:
\begin{lemma}
The Fourier transform of a 1D normally distributed signal of size d with standard deviation $\sigma$ is normally distributed with standard deviation $\sigma/\sqrt{2d}$
\begin{equation}
        \mathcal{F}[G_\sigma]_{k} = \left[G_{\sigma/\sqrt{2d}}\right]_{k} + i \left[G_{\sigma/\sqrt{2d}}\right]_{k}
\end{equation}
For an $m$-dim normal distributed signal of over a grid of volume $d^m$, the same holds with a standard deviation of $\sigma/\sqrt{2d^m}$. 
\end{lemma}
\begin{proof}
Focus on the real part of $\mathcal{F}[G_\sigma]_{k}$ (the result for the imaginary part follows by symmetry)
\begin{equation}
        \Re\,\mathcal{F}[G_\sigma]_{k} = \frac{1}{d} \sum_{j=0}^d x_j \cos(2\pi\frac{kj}{d})
\end{equation}
First first notice that $\langle\Re\,\mathcal{F}[G_\sigma]_{\boldsymbol k}\rangle = 0$.
Now consider the variance
\begin{equation}
        \sigma^2(\Re\,\mathcal{F}[G_\sigma]_{k}) = \frac{1}{d^2} \sum_{jj'} \langle x_j x_{j'} \rangle \cos(2\pi\frac{kj}{d})\cos(2\pi\frac{kj'}{d})    
\end{equation}
Since each $x_m$ is independently distributed, $\langle x_j x_{j'} \rangle$ is only nonzero when $j=j'$
\begin{align}
        \sigma^2(\Re\,\mathcal{F}[G_\sigma]_{k}) &= \frac{1}{d^2} \sum_{j} \sigma^2 \cos^2(2\pi\frac{kj}{d}) \\
                                                &= \frac{ \sigma^2}{2d}.
\end{align}
All other moments are zero as the higher moments of a normal distribution are zero. Therefore, the distribution is normal with mean 0 and standard deviation $\frac{\sigma}{\sqrt{2d}}$. A similar argument leads to the result for higher dimensions.
\end{proof}
Finally, it remains to be shown that the convolution of a normally distributed signal and an arbitrary signal is still normally distributed. We can show
\begin{lemma}
Take an arbitrary deterministic signal $\{y_k\}$. Then $[G_{\sigma} * y]_k$ is a normally distributed signal with mean 0 and standard deviation $\sigma||y||_2$.
Note that this time, the distributions at different $k$'s are no longer necessarily independent.
\end{lemma}
\begin{proof}
By definition of the convolution,
\begin{equation}
[G_{\sigma} * y]_k = \sum_l [G_{\sigma}]_l y_{k-l}
\end{equation}
Now we have $\langle [G_{\sigma} * Y]_k \rangle = 0$ as the normal distribution has zero mean. For the variance we can use a similar argument as before
\begin{align}
        \sigma^2([G_{\sigma} * y]_k) &= \sum_{ll'} \langle[G_{\sigma}]_l [G_{\sigma}]_{l'}\rangle y_{k-l} y_{k-l'} \\
                                                                        &= \sum_{l} \sigma^2 y^2_{k-l} = \sigma^2 ||y||_2^2
\end{align}
As before, all the higher moments vanish, implying a normal distribution.\\
To see that the distributions are no longer independent in general, consider the two-point correlation
\begin{align}
        \langle [G_{\sigma} * y]_k [G_{\sigma} * y]_k' \rangle &= \sum_{ll'} \langle[G_{\sigma}]_l [G_{\sigma}]_{l'}\rangle y_{k-l} y_{k'-l'} \\
                                                                        &= \sigma^2 \sum_{l} y_{k-l} y_{k'-l}
\end{align}
Therefore the correlation between the distributions at $k, k'$ depends on the autocorrelation characteristics of the coefficients $y_k$. 
\end{proof}
Combining the results of these two lemmas, one deduces that each Fourier coefficient $\mathcal{F}[\tilde{x}]_{\boldsymbol{k}}$ is normally distributed around the noiseless mean $\mathcal{F}[x]_{\boldsymbol{k}}$ with standard deviation
\begin{align}
        \sigma(\Re\,\mathcal{F}[\tilde{x}]_{\boldsymbol{k}})    &= \frac{||\mathcal{F}[\sqrt{1 - x^2}]||_2}{\sqrt{2 n_s d^m}} \\ 
                                                                                                                        &= \frac{\sqrt{d^m - ||x||_2^2}}{d^m\sqrt{2 n_s}}
\end{align}
where in the second equality we used Parseval's identity. The imaginary part is identically and independently distributed because of symmetry.
\end{proof}

\section{Figures of merit}
\label{sec:figures-merit}

\begin{lemma}[Average fidelity]
The average fidelity of the output state over all parameters is proportional to the inner product of the vectors of Fourier coefficients:
\begin{equation}
        \langle F\rangle_{\boldsymbol\theta} = \frac{1}{2^n} \vec c\,^\dagger \vec{\tilde c}
\end{equation}
\end{lemma}
\begin{proof}
As a reminder, we use the vectorised notation $\vec c$ for the Fourier coefficients over all Paulis: $[\vec c]_{\boldsymbol k, i} = c_{\boldsymbol k}(P_i)$.
Consider the following expression for the fidelity between a pure state and an arbitrary state \cite{flammia2011direct}
\begin{equation}
        F(\sigma, \rho) = \sum_{i} \text{Tr}[\chi_i \sigma]\text{Tr}[\chi_i \rho]
\end{equation}
where the summation is over some informationally complete set of observables, which we can take to be the normalised Pauli observables: $\chi_i = P_i/\sqrt{2^n}$.
Now consider the fidelity between a state produced by a noisy parameterised quantum circuit $\tilde\rho(\boldsymbol\theta)$ and its exact counterpart $\rho(\boldsymbol\theta)$.

In the Fourier picture, the coefficients of the fidelity as a function of $\boldsymbol\theta$ are given by a convolution between the Fourier coefficients of the expectation values of the Pauli observables
\begin{equation}
        \mathcal{F}[F]_{\boldsymbol k} = \frac{1}{2^n} \sum_i \left[c(P_i) * \tilde{c}(P_i)\right]_{\boldsymbol k}
\end{equation}
where the discrete convolution is defined as:
\begin{equation}
        \left[c * \tilde{c}\right]_{\boldsymbol k} = \sum_{\boldsymbol x\in \boldsymbol \Lambda} c_{\boldsymbol x} \tilde{c}_{\boldsymbol k-\boldsymbol x}
\end{equation}
Therefore, the set of Fourier coefficients over all Pauli observables determines the fidelity as a function of the parameters.

If we define the average fidelity over all sampled parameters:
\begin{equation}
        \langle F\rangle_{\boldsymbol\theta} : = \frac{1}{d^m} \sum_{\boldsymbol{j}} F(\rho(\boldsymbol{\theta}_{\boldsymbol{j}}), \tilde{\rho}(\boldsymbol{\theta}_{\boldsymbol{j}}) )
\end{equation}
we see that it is equivalent to $\mathcal{F}[F]_{\boldsymbol 0}$. Hence:
\begin{equation}
        \langle F\rangle_{\boldsymbol\theta} = \frac{1}{2^n} \sum_i \sum_{\boldsymbol x\in \boldsymbol \Lambda} c_{\boldsymbol x}(P_i) \tilde{c}_{-\boldsymbol x}(P_i)
        = \frac{1}{2^n} \sum_i \sum_{\boldsymbol x\in \boldsymbol \Lambda} c_{\boldsymbol x}(P_i) \tilde{c}^{*}_{\boldsymbol x}(P_i) = \frac{1}{2^n} \vec c\,^\dagger \vec{\tilde c}
\end{equation}
which concludes the proof.
\end{proof}
As a remark, for the average fidelity defined as above to correspond to the average over all possible parameters $\int F(\rho(\boldsymbol{\theta}), \tilde{\rho}(\boldsymbol{\theta}) ) d^m\boldsymbol{\theta}$ a sufficient condition is for the sampling rate to be twice the maximum frequency of the signal, via Nyquist-Shannon sampling theorem.

We also have a corresponding lemma for the purity:
\begin{lemma}[Average purity]
The average purity of the output state over all parameters is proportional to the squared norm of the vector of Fourier coefficients:
\begin{equation}
        \langle P\rangle_{\boldsymbol\theta} = \frac{1}{2^n} |\vec c|^2
\end{equation}
\end{lemma}

\begin{proof}
Notice that in the definition of fidelity above, $P(\rho) = F(\rho, \rho)$. The result follows therefore by taking $\tilde{c} = c$ in the previous Lemma.
\end{proof}

\section{Filtering-based error mitigation}
\label{sec:AppendixFilter}

\subsection{Perfect filtering for classically simulable circuits}
\label{sec:perf-filt-class}

\begin{theorem}
        Take a noisy parameterised quantum state $\tilde\rho(\boldsymbol\theta) = \tilde\U(\boldsymbol\theta)(\rho_0)$ such that all the unparameterised unitaries in $\U(\boldsymbol\theta)$ are Clifford and $\rho_0$ is stabiliser. Then given a Pauli measurement operator $O$, there exists a noise threshold, under which the noiseless landscape $\langle O \rangle_{\rho(\boldsymbol\theta)}$ can be perfectly reconstructed by sampling from the noisy landscape $\langle O \rangle_{\tilde{\rho}(\boldsymbol\theta)}$. Above the threshold, such reconstruction is not possible.
\end{theorem}
\begin{proof}
        In Section \ref{sec:Example} we saw that the noiseless cost function, when expressed as a trigonometric polynomial, has coefficients $c'_{\boldsymbol k}\in\{0, \pm 1\}$ with support $\boldsymbol\Lambda$ as defined in Theorem \ref{thm:1}. We also saw in Section \ref{sec:IIIA} that noise channels introduce a perturbation in these coefficients, such that under noise they become $\tilde c'_{\boldsymbol k}$ with potentially different support $\tilde{\boldsymbol{\Lambda}}$. Therefore, given perfect access to the noisy cost function, as long as $|\tilde c'_{\boldsymbol k} - c'_{\boldsymbol k}| \le \frac{1}{2} \,\, \forall \boldsymbol k \in \boldsymbol\Lambda$, a thresholding algorithm can perfectly reconstruct the noiseless trigonometric coefficients by 
        \begin{equation}
                S(\tilde c'_{\boldsymbol k}) = \begin{cases}
                        \text{round}(\tilde c'_{\boldsymbol k}) \; &\text{if} \; \boldsymbol k \in \boldsymbol \Lambda ; \\
                        0 \; &\text{otherwise} .
                \end{cases}
        \end{equation}
        Conversely, if noise is larger than the threshold, then it is not guaranteed that the outcome will be the correct noiseless landscape. This is a consequence of the fact that parameterised quantum circuits are universal approximators, as proven in \cite{schuld2021effect}. Therefore a landscape produced by a highly noisy state could be corrected to one produced by another, incorrect, circuit.
\end{proof}
In Figure \ref{fig:reconstruction} we show a practical example of reconstructing such a landscape. The landscape is sampled at different numbers of shots and therefore different levels of noise, showing the existence of a well-defined noise threshold for perfect reconstruction.

\begin{figure}[ht!]
        \includegraphics[width=0.6\textwidth,center]{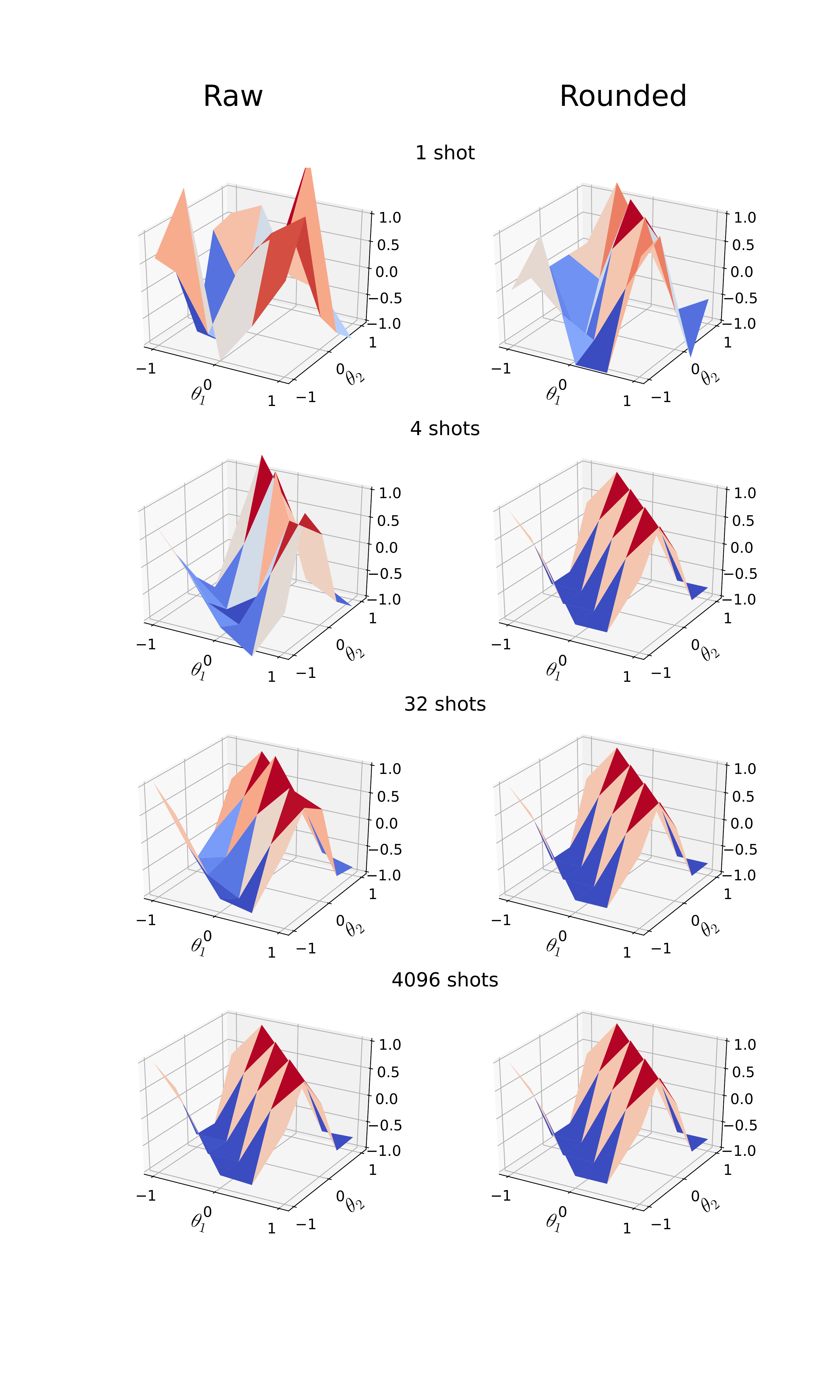}
        \caption{Reconstructing the landscape of a classically simulable circuit as run on Quantinuum H1-1 system, at different numbers of shots. Left is raw landscape, right is after rounding of the trigonometric coefficients. Colour indicates the height. Perfect reconstruction is obtained for all experiments except the 1-shot one, the exact (noiseless) landscape corresponds to the last three diagrams in the right column. }
        \label{fig:reconstruction}
\end{figure}

\subsection{Filters for noise-induced frequency modes}
\label{sec:filt-noise-induc}

In line with classical signal processing, some of the simplest filters that can be applied are band-pass filters, that allow particular frequencies while eliminating the others. These can be designed based on knowledge of the theoretical frequency spectrum, which in the most general case is given by Theorem~\ref{thm:1}. 
In general, given an observed function $\tilde f$ of quantum expectation values on a support spanned by a basis $\{\psi_{\boldsymbol k}\}$ (which may be the Fourier basis or any other basis of choice) indexed by a lattice $\tilde{\boldsymbol{\Lambda}}$, as well as a known exact support indexed by a lattice $\boldsymbol \Lambda \subset \tilde{\boldsymbol{\Lambda}}$, the denoised function is given by
\begin{equation}
        f^\# = \sum_{\boldsymbol k \in \boldsymbol \Lambda} \langle f, \psi_{\boldsymbol k} \rangle \psi_{\boldsymbol k} 
\end{equation}
In other words, such a filter cuts off all frequencies beyond those expected to be output by the circuit. Such a method would be especially helpful to mitigate correlated noise that can give rise to high-frequency modes. However, it would most likely not be helpful for other non-correlated noise channels, as well as being more resource-intensive by requiring the resolution of a larger set of frequencies.\\
For specific circuits, it is possible to deduce additional characteristics that the noiseless spectrum should have, allowing for an improved noise mitigation.
For example, take the case where the measurement operator is local, i.e. acts on only a subset of qubits. Then it is clear that its expectation value can only depend on the parameters in its backwards light-cone. 
For a more involved example, assume the input state $\rho_0$ is known, and the parameters of interest are those of the first $r$ unitaries, which are Pauli gadgets $\U_i(\theta_i) = e^{-iP_i\theta_i/2}$, $i = 1, 2, \cdots, r$. Further, suppose that all $P_i$'s commute, such that the ordering of the Pauli gadgets does not matter. All the subsequent gates, as well as the measurement operator, are arbitrary. Then, if $[\rho_0, P_i] \neq 0$ for all $i \le r$, it follows that $f$ must contain terms where either all parameters $\theta_1, \cdots, \theta_r$ appear, or none of them do. This is because 
 a parameterised Pauli gadget that does not commute with its input state gives an output state that is a non-constant function of that parameter. A similar reasoning can be applied on the measurement operator and by extension to the expectation value. These conditions apply in many situations of practical interests, thereby affording a more refined error mitigation than general circuits.

\subsection{Thresholding}
\label{sec:IIIB}
The previous filtering methods relied on specific knowledge of the circuits. It is desirable to devise mitigation methods that would apply in general, and therefore would need to be agnostic to the circuit structure. One such method is based on isolating the most significant coefficients. In classical signal processing, this is known as thresholding, and was first studied in Ref.~\cite{donoho1994ideal}.
As a basic review of the method, thresholding applies a non-linear transform to the coefficients of some basis
\begin{equation}
        f^\# = \sum_{\boldsymbol k \in \tilde{\boldsymbol \Lambda} } S(\langle f, \psi_{\boldsymbol k} \rangle) \psi_{\boldsymbol k}
\end{equation}
There exist two common choices for the transform $S$, known as `hard' and `soft' thresholding. Hard thresholding has
$S^H_T(x) = \begin{cases}
x &\text{if } |x| > T \\
0 &\text{if } |x| \le T
\end{cases}$, while soft thresholding has
$S^S_T(x) = \max (1-T/|x|, 0)x$~\cite{mallat2009wavelet}.\\
In~\cite{donoho1994ideal} it was proven that, given an $N$-sparse signal perturbed by additive Gaussian noise with mean 0 and standard deviation $\sigma$, the optimal threshold is $T = \sigma\sqrt{2\log N}$, however in practice it is observed that $T = O(\sigma)$ is sufficient for both hard and soft thresholding~\cite{peyre2021mathematical}.

Thresholding methods rely on the basic assumption that the noiseless signal is sparse, or approximately sparse, in the basis of interest. As this in general is not the case for quantum algorithms, we only claim that thresholding be useful on a heuristic level. Another assumption that might not hold is that of additive Gaussian white noise. In the Fourier basis, this is only true for shot noise, however quantum noise will not generally obey these requirements. At the same time, from an experimental point of view, it does make at least practical sense to assume Gaussian noise to reflect the lack of knowledge of these noise processes. To this end, estimates of $\sigma$ could be derived from a combination of the predicted sources of error, including known decoherence and measurement error rates.

\end{document}